\documentclass[12pt,reqno]{amsart}
\usepackage[foot]{amsaddr}

\usepackage{amsmath,amssymb,amsfonts}

\usepackage{url}       %
\usepackage{xcolor}
\usepackage{soul}
\usepackage[margin=1in]{geometry}
\usepackage{graphicx}
\usepackage{amsthm}
\usepackage{setspace}
\usepackage{natbib}
\usepackage{todonotes}
\usepackage{amsthm}
\usepackage{geometry}
\usepackage{graphicx}
\usepackage{multirow}
\usepackage{caption}
\usepackage{subcaption}
\usepackage{tabularx}
\usepackage[flushleft]{threeparttable}
\usepackage{array}
\usepackage{booktabs}
\usepackage{pdflscape}
\usepackage{lscape}
\usepackage{bbm}
\usepackage{color}
\usepackage{blkarray}

\usepackage{floatrow}

\newcommand{\vertii}[1]{{\left\vert\kern-0.25ex\left\vert #1 
		\right\vert\kern-0.25ex\right\vert}}

\newcommand{\vertiii}[1]{{\left\vert\kern-0.25ex\left\vert\kern-0.25ex\left\vert #1 
		\right\vert\kern-0.25ex\right\vert\kern-0.25ex\right\vert}}

\newtheorem{theorem}{Theorem}[section]

\newtheorem{assumption}{Assumption}[section]

\newtheorem{lemma}{Lemma}[section]

\newtheorem{proposition}{Proposition}[section]

\usepackage{soul} 

\usepackage{chngcntr}
\usepackage{apptools}
\AtAppendix{\numberwithin{theorem}{subsection}}
\AtAppendix{\numberwithin{equation}{subsection}}

\title[Estimation of High-Dimensional Seemingly Unrelated Regression Models]
{\bf Estimation of High-Dimensional Seemingly Unrelated Regression Models}

\date{\today}

\author{Lidan Tan}
\address{Tan: Department of Economics, University of Southern California}
\author{Khai X. Chiong} 
\address{Chiong: Naveen Jindal School of Management, University of Texas at Dallas}
\author{Hyungsik Roger Moon} \thanks{Tan: lidantan@usc.edu, Chiong: Khai.Chiong@utdallas.edu, Moon: moonr@usc.edu, corresponding author.}
\address{Moon: Department of Economics, University of Southern California and School of Economics, Yonsei University} 

\begin{document}
	
	\begin{abstract}
	In this paper, we investigate seemingly unrelated regression (SUR) models that allow the number of equations ($N$) to be large, and to be comparable to the number of the observations in each equation ($T$). It is well known in the literature that the conventional SUR estimator, for example, the generalized least squares (GLS) estimator of \cite{zellner1962efficient} does not perform well. As the main contribution of the paper, we propose a new feasible GLS estimator called the feasible graphical lasso (FGLasso) estimator. For a feasible implementation of the GLS estimator, we use the graphical lasso estimation of the precision matrix (the inverse of the covariance matrix of the equation system errors) assuming that the underlying unknown precision matrix is sparse. We derive asymptotic theories of the new estimator and investigate its finite sample properties via Monte-Carlo simulations.
	
	\bigskip
	
\end{abstract}
		
\maketitle

	\section{Introduction}
	
	A SUR comprises multiple individual regression equations that are correlated with each other. In our setup, we assume that there are $N$ regression equations that are observed over periods $t=1,2,...,T$. These regression equations are related in the sense that the regression errors of the equation system are correlated.
	
	The SUR estimator originally proposed by \cite{zellner1962efficient} is a feasible generalized least squares (FGLS) estimator that is based on an estimator of the inverse of the covariance matrix, the precision matrix ($\Omega := \Sigma^{-1}$), of the SUR equation system. Often this estimator is computed in two steps. In the first step, one estimates each equation by the ordinary least squares (OLS) and computes the residuals. In the second step, one computes the FGLS based on the inverse of the sample covariance matrix of the residuals.
	
	It is well-known that when the number of equations, $N$, is large, the FGLS estimator performs poorly (e.g., \cite{greene2003econometric}). 	The main reason is when $N$ is large relative to the number of observations $T$, the sample precision matrix (inverting the sample covariance matrix) performs poorly.   For example, when $N/T \to c>0$, the empirical covariance matrix is not consistent and would be rank deficient when $N>T$ (e.g., see \cite{johnstone2001distribution} and \cite{hastie2015statistical}).  
	In this paper, we revisit the problem of estimating the classical SUR model when the number of regression equations, $N$, is large and even comparable to $T$.
	
	The problem of estimating high dimensional precision matrices has been widely studied in the machine learning and statistical learning literature.  A popular method is to estimate the precision matrix with various regularizations -- imposing various restrictions that the precision matrix is sparse. (e.g., see \cite{cai2011constrained}, \cite{negahban2011estimation}, \cite{friedman2008sparse}, \cite{lam2009sparsistency}.) 
	
	It is well known that the sparsity of the precision matrix has a nice interpretation if the underlying distribution is Gaussian - the set of non-zero entries in the precision matrix correspond to the set of edges in an associated Gaussian Markov random field (GMRF) (see \cite{hastie2015statistical}). Thus imposing sparsity on the precision matrix corresponds to the assumption that not all regression equations are related to each other. In another words, the graph representing which regression equation is related to which, is sparse.\footnote{\cite{liu2009nonparanormal} relaxed the results to a more general class of distribution named as non-paranormal distribution.} 	 This restriction is reasonable -- a prevailing result in the literature of social and economic networks is that these graphs are sparse, and the degrees of the nodes  grow much slower than the network size $N$ (\cite{barabasi2016network}).
	
	In this paper, assuming sparsity of the true precision matrix $\Omega$, we propose a new estimator called the FGLasso that works when both $N,T \to \infty$ and under some conditions, $N \geq T$.  Unlike FGLS which estimates $\Sigma$ by the OLS residuals and then taking the inverse, we directly estimate the precision matrix $\Omega:=\Sigma^{-1}$ using valid high-dimensional techniques. In particular, we directly estimate the  precision matrix using the  {\em Graphical Lasso} estimator, which is a popular estimator for the  high-dimensional precision matrix (see \cite{fan2016overview}).

Our results are as follows. We find a set of regularity conditions under which our FGLasso estimator is asymptotically equivalent to the (infeasible) GLS estimator uniformly across equations. Moreover, we  show that if the maximum nonzero entries per row in $\Omega$ is bounded or grows much smaller than $N$, FGLasso estimator performs well even when $N > T$. In the Monte-Carlo study, we compare the performance of the FGLasso estimator with the OLS estimator, the GLS estimator, and the FGLS estimator, corroborating our findings that our proposed estimator performs better in a high-dimensional setting. In deriving these results, we build upon \cite{ravikumar2011high}, which show that under certain regularity conditions, the graphical lasso estimator $\widehat{\Omega}_{gl}$ converges to the true $\Omega$ at the rate $\mathcal{O}_p(\sqrt{\log N/T})$, in terms of the element-wise maximum norm, while preserving the sparsity pattern.
	
	The remainder of this paper is organized as follows. Section 2 discusses the SUR model in details, summarizes the OLS, the GLS, the FGLS, and the FGLasso estimators. In section 3, we discuss the results from \cite{ravikumar2011high} and present the main theoretical results. Section 4 reports Monte Carlo simulation results\footnote{Computation for the work described in this paper was supported by the University of Southern California’s Center for High-Performance Computing (hpc.usc.edu).} and section 5 concludes. All the technical proofs and additional simulation results are provided in the appendix. 
	
	\section*{Notation}
	For the convenience, we briefly summarize the notation to be used throughout the paper. We denote $s_{\mathrm{min}}(A)$ and $s_{\mathrm{max}}(A)$ as the maximum and the minimum singular values of real valued matrix $A\in \mathbb{R}^{m \times n}$, respectively. The operator norm and Frobenius norm are defined as $\vertii{A}_{op} = s_{\mathrm{max}}(A)$ and  $\vertii{A}_F=\sqrt{\sum_{i,j}A_{ij}^2}$, respectively. Let $\vertii{A}_\infty = \underset{i,j}{\mathrm{max}}|A_{i,j}|$ denote the element-wise maximum norm and  $\vertiii{A}_\infty = \underset{i=1,2,...,L}{\mathrm{max}}\sum_{j=1}^{N'}|A_{ij}|$the maximum absolute row sum matrix norm. Let $A'$ denote the transpose of $A$ and $\otimes$ denote the Kronecker product. For a real sequence $\{a_n\}_{n=1}^\infty$ and a positive sequence $\{b_n\}_{n=1}^\infty$, we denote $a_n=\mathcal{O}(b_n)$ if there exists a finite constant $C$ such that $|a_n| \leq C b_n$ as $n \to \infty$, and $a_n = \mathcal{O}_p(b_n)$ if $\mathbb{P}(|a_n| \leq C b_n) \to 1$ as $n \to \infty$. We use notation $\Rightarrow$ and $\xrightarrow{p}$ to denote the convergence in distribution and the convergence in probability, respectively.

	\section{Setup}
	
	\subsection{SUR model}
	 Suppose we estimate a system of linear equations: 
	\begin{equation}	\label{m.single equation}
	Y_{it} = \beta_i' X_{it} + U_{it}
	\end{equation}
	for $i = 1,\cdots,N$ and $t=1,\cdots,T$. Here $X_{it}=(X_{it,1},X_{it,2},\cdots,X_{it,K_i})'$ is a $K_i$- column vector of the regressors for unit $i$, and $U_{it}$ is the unobserved error term. The heterogeneous regression coefficients $\beta_i \in \mathbb{R}^{K_i \times 1}$ are the parameters of interest. 
	
	Stacking the observations over $N$ units, let $Y_t = (Y_{1t},...,Y_{Nt})' \in \mathbb{R}^N $, $U_t = (U_{1t},...,U_{Nt})' \in \mathbb{R}^N $, $X_t = diag(X_{1t},...,X_{Nt}) \in \mathbb{R}^{N \times \sum_{i=1}^N K_i}$, and $\beta = (\beta_1',...,\beta_N')' \in \mathbb{R}^{\sum_{i=1}^N K_i}$. The system of the equations in (\ref{m.single equation}) can be expressed as 
	\begin{equation}	\label{m.system}
	Y_t = X_t' \beta + U_t
	\end{equation}  
	Alternatively, stacking the observations in (\ref{m.single equation}) over $t$, we can also express the system of the equations in (\ref{m.single equation}) as
	\begin{equation}
	Y_i=X_i \beta_i+U_i,	\label{m.stackoverT}
	\end{equation}
	where $Y_i=(Y_{i1},Y_{i2},\cdots,Y_{iT})'\in \mathbb{R}^T$, $X_i=(X_{i1},X_{i2},\cdots,X_{iT})'\in \mathbb{R}^{T \times K}$ and $U_i=(U_{i1},U_{i2},\cdots,U_{iT})'\in \mathbb{R}^T$. 
	
	In a matrix form, we can write the model as
	\begin{equation}
	Y=X \beta+U,	\label{m.fullequation}
	\end{equation}
	where $Y=(Y_1',Y_2',\cdots,Y_N')' \in \mathbb{R}^{NT}$, $U=(U_1', U_2',\cdots, U_N')'\in \mathbb{R}^{NT}$, \[X=\begin{bmatrix}
	X_{1} & 0  & 0 & \dots  & 0 \\
	0 & X_{2} & 0 & \dots  & 0 \\
	\vdots & \vdots & \vdots & \ddots & \vdots \\
	0 & 0 & 0 & \dots  & X_{N}
	\end{bmatrix}_{TN \times \sum_{i=1}^N K_i} \text{and} \quad \beta=\begin{bmatrix}
	\beta_1 \\
	\beta_2 \\
	\vdots \\
	\beta_N
	\end{bmatrix} .\]
	
	In this paper, we assume the classical linear system equation assumptions:
	\begin{assumption}[\bf Model]	\label{assumption1} We assume:
		\begin{itemize}
			\item[(i)] $X$ is a full rank matrix, and at least there exists a pair $(i,j)$ such that $X_{i} \neq X_{j}$. 
			\item[(ii)] $E(U_t|X_1,\cdots,X_T)=0$.
			\item[(iii)] $E(U_tU_t'|X_1,\cdots,X_T)=\Sigma>0$.
		\end{itemize}
	\end{assumption}
	
	The conditions in Assumption \ref{assumption1} are quite classical in the SUR literature. The first condition excludes the case all the regressors are identical. In this case, it is well known that the OLS estimator becomes efficient and there is no gain of using the information in the equation system. The second condition assumes that the regressors are strictly exogenous and it excludes regressors like lagged dependent variables. This assumption may be restrictive in many applications, but we need this assumption as a technical regularity condition that is required for proving the asymptotic properties of $GLS$ estimator and $FGLasso$ estimator when $N,T\to \infty$ (see details in Appendix A.). The third condition is the homoskedasticity assumption. This condition may be relaxed to allow for conditional heteroskedasticity but at the cost of technical complexity of the asymptotic results of the paper. We assume it just for simplicity in deriving asymptotic results.

	Without loss of generality, in the rest of the paper, we assume the number of regressors of each unit is the same, i.e. $K_i=K,$ $i=1,2,\hdots,N$, a constant number. The case in which $K_i's$ are different can be easily extended. The results remain the same as long as $K_i=\mathcal{O}(1)$ for all $i$.
	
	\subsection{Estimators}
	In this section, we first briefly summarize the OLS, the GLS, and the FGLS estimators of $\beta$ in the SUR model. Then we introduce the FGLasso estimator. 
	
	The OLS estimator is defined as
	\begin{equation} \label{e.ols}
	\widehat{\beta}_{OLS}=\Big(\sum_{i=1}^TX_tX_t'\Big)^{-1}\sum_{t=1}^TX_tY_t .
	\end{equation}
	It is equivalent to the OLS estimators of individual equations,\\  $\widehat{\beta}_{OLS}=(\widehat{\beta}_{1,OLS}',\widehat{\beta}_{2,OLS}',\cdots,\widehat{\beta}_{N,OLS}')'$, where $\widehat{\beta}_{i,OLS}=\big(X_i' X_i\big)^{-1} \big(X_i' Y_i\big)$ for $i=1,2,...,N$.
	
	In his seminar paper, \cite{zellner1962efficient} proposed the SUR estimator to improve the OLS estimator by exploiting the correlation in the equation system. Suppose that $\Sigma$ is known. As earlier, define the precision matrix as $\Omega:=\Sigma^{-1}$. Then GLS estimator is defined as
	\begin{equation} \label{e.gls}
	\widehat{\beta}_{GLS}=\Big(\sum_{t=1}^TX_t\Omega X_t'\Big)^{-1}\sum_{t=1}^TX_t\Omega Y_t.
	\end{equation}
	In most applications, however, $\Sigma$ and $\Omega$ are not known. A FGLS estimator (see details in \cite{greene2003econometric}) is defined by replacing the unknown $\Sigma$ with the consistent estimator. A widely used estimator of $\Sigma$ is $\widehat{\Sigma}=\frac{1}{T}\sum_{t=1}^T\widehat{U}_t\widehat{U}_t'$ and $\widehat{U}_t$ is the $OLS$ residuals, that is, $\widehat{U}_t=Y_t-X_t'\widehat{\beta}_{OLS}$. Then
	\begin{equation} 	\label{e.fgls}
	\widehat{\beta}_{FGLS}=\Big(\sum_{t=1}^TX_t\widehat{\Sigma}^{-1}X_t'\Big)^{-1}\sum_{t=1}^TX_t\widehat{\Sigma}^{-1} Y_t.
	\end{equation}
	
	The FLGS estimator in (\ref{e.fgls}) suffers from two major problems when $N$ is large.
	Suppose that $T>N$, but both $T$ and $N$ are large and in the same order. It is known that the estimator $\widehat{\Omega}=\widehat{\Sigma}^{-1}$ behaves poorly. Further, $\widehat{\Sigma}^{-1}$ is only well defined when $T \geq N$. When $T$ is less than $N$, $\widehat{\Sigma}$ is rank deficient and therefore not invertible. 
	
	Our estimator is motivated by these two issues. Suppose that $\Omega$ is sparse. In this case, we propose FGLasso estimator by replacing $\widehat{\Omega} = \widehat{\Sigma}^{-1}$ in (\ref{e.fgls}) with a graphical lasso estimator, $\widehat{\Omega}_{gl}$, where
	\begin{equation}	 \label{e.gl}
	\widehat{\Omega}_{gl}=\underset{\Omega>0, \Omega^T=\Omega} {\mathrm{argmin}}\big\{\mathrm{tr}(\Omega\widehat{\Sigma})-\mathrm{logdet}\Omega+\lambda_n \vertii{\Omega}_{1,\textrm{off}}\big\},
	\end{equation}
	$\vertii{\Omega}_{1,\textrm{off}}=\sum_{i\neq j}^N |\Omega_{ij}|$ and $\lambda_n>0$ is a penalization parameter which is often chosen by a cross-validation method (e.g., see \cite{friedman2008sparse}). More specifically,  
	\begin{equation} \label{e.fglasso}
	\widehat{\beta}_{FGLasso}= \left( \sum_{t=1}^TX_t\widehat{\Omega}_{gl}X_t'\right)^{-1}
	\left( \sum_{t=1}^TX_t\widehat{\Omega}_{gl}Y_t \right).
	\end{equation}
	
	In the next section, under a certain restriction on the sparse structure of $\Omega$, we show that $\widehat{\Omega}_{gl}$ is consistent even when $T\leq N$, and $\widehat{\beta}_{FGLasso}$ is consistent and has similar asymptotic properties as $\widehat{\beta}_{GLS}$.

\section{Asymptotic Properties of FGLasso Estimator}

	The sample properties of the FGLasso estimator is mainly dependent on the sample properties of $\widehat{\Omega}_{gl}$ in (\ref{e.gl}). Intuitively, if $\widehat{\Omega}_{gl}$ is close to true $\Omega$ in some metric, then $\widehat{\beta}_{FGLasso}$ would also be close to $\widehat{\beta}_{GLS}$. In section 3.1, we briefly summarize results in \cite{ravikumar2011high} regarding the properties of $\widehat{\Omega}_{gl}$. In section 3.2, we show the consistency and asymptotic properties of $\widehat{\beta}_{FGLasso}$.
	
	\subsection{Properties of $\widehat{\Omega}_{gl}$}
	
	
		Define $\Gamma^*=\Omega \otimes \Omega$ and let $E(\Omega)$ be the edge set including self-link, i.e $E(\Omega):=\{(i,j)|\Omega_{ij}\neq 0\}$. Denote $S^c$ be the complement of S set. For set $P,P'$, we also use $\Gamma^*_{PP'}$ to denote the $|P|\times |P'|$ matrix with rows and columns of $\Gamma^*$ indexed by $P$  and $P'$ respectively. $\kappa_{\Sigma}:=\vertiii{\Sigma}_\infty$ and $\kappa_{\Gamma}:=\vertiii{(\Gamma^*_{SS})^{-1}}_\infty$.
	\begin{assumption} \label{p.assumption2} ~\ We assume 
		\begin{itemize}
			\item[(i)] There are at most $D_N$ nonzero entries per row in $\Omega$.
			\item[(ii)] Conditional on $X$, $U_{it}/\Sigma_{ii}$ is i.i.d sub-Gaussian over $t$ with parameter $\sigma$, i.e.
			\begin{equation}	\label{p.def_subGaussian}
			\mathbb{E}(e^{\lambda (U_{it}/\Sigma_{ii})})\leq e^{\lambda^2 \sigma^2}, \quad \text{for all }\lambda\in \mathbb{R}.
			\end{equation}
			\item[(iii)] There exists some $\alpha \in (0,1]$ such that
			\begin{equation}  \label{p.incohereance condition}
			\underset{e\in S^c}{\mathrm{max}}\vertii{\Gamma^*_{eS}(\Gamma^*_{SS})^{-1}}_1\leq 1-\alpha.
			\end{equation}
			\item[(iv)] $(\kappa_{\Gamma^*}, \kappa_{\Sigma}, \alpha,\sigma)$ remain constant as a function of $(N,T)$.
		\end{itemize}
	\end{assumption}
	
	 The first condition assumes that $D_N$ is an upper bound of the number of nonzero elements in $\Omega$. It usually refers to the maximum degree of nodes in graph theory literature. In this paper, we do not put a restriction that a finite constants bounds $D_N$. We allow $D_N$ to increase slowly to infinity as $N$ increases. The second condition assumes that the distribution of $U_{it}$ has a thin tail like the Gaussian distribution. The third assumption is usually referred as the incoherence condition, which guarantees the exact recovery of $\widehat{\Omega}_{gl}$. The condition that $\kappa_\Sigma$ remains constant in the third assumption indicates that the singular value of true precision matrix $\Omega$ is lower bounded\footnote{Note $s_{min}(\Omega)=\frac{1}{s_{max}(\Sigma)}\geq \frac{1}{\vertiii{\Sigma}_\infty}$.}. The assumption that the rest of the parameters $(\kappa_{\Gamma^*},\alpha,\sigma)$ remain constant is only for simplicity, more detailed results and discussion can be found in Corollary 1 in \cite{ravikumar2011high}.

	\begin{lemma}[\bf \cite{ravikumar2011high}] \label{lemma.ravikumar}
		Assume Assumption \ref{p.assumption2} holds. If $T \geq c D_N^2 \log N$ for constant $c>0$, then the optimal solution $\widehat{\Omega}_{gl}$ in (\ref{e.gl}) satisfies:
		\begin{itemize}
			\item[(i)] \begin{equation} \label{p.ravikumar}
			\vertii{\widehat{\Omega}_{gl}-\Omega}_\infty=\mathcal{O}_p\left(\sqrt{\frac{\log N}{T}}\right);
			\end{equation}
			
			\item[(ii)] The edge set $E(\widehat{\Omega}_{gl})$ is a subset of the true edge set $E(\Omega)$, and includes all edges $(i,j)$ with $|\Omega^*|>c'\sqrt{\frac{\log N}{T}}$, where $c'>0$ is a small constant number that depends on $\sigma, \alpha, \kappa_\Gamma^*$ and $\mathrm{max}_i(\Sigma_{\mathrm{ii}})$.
		\end{itemize}
	\end{lemma}

%
	
	The first result of Lemma \ref{lemma.ravikumar} guarantees that, the error between each element of $\widehat{\Omega}_{gl}$ and $\Omega$ shrinks uniformly at a rate of $\sqrt{\log N/T}$, meaning that as long as $T$ increases faster than $\log N$, the error will go to zero. The maximum degree $D_N$ plays an important role as it determines the lower bound of sample size $T$. If $D_N$ is bounded or increases much slower than $N$ (for example, $D_N=\log N$), then it is possible that the properties in Lemma \ref{lemma.ravikumar} hold for $T<N$.
	
	 The second result of Lemma \ref{lemma.ravikumar} shows the exact recovery property, meaning that $\widehat{\Omega}_{gl}$ from (\ref{e.gl}) remains similar sparsity structure of $\Omega$. The non-edge set $E^c(\Omega) \subseteq E^c(\widehat{\Omega}_{gl})$, that is, if $\Omega_{ij}=0$, then $\widehat{\Omega}_{gl,ij}=0$ wp1. Therefore, consider matrix $\Delta_\Omega:=\Omega-\widehat{\Omega}_{gl}$, the maximum nonzero entries per row is at least $D_N$.
	
	\subsection{Asymptotic Properties of $\widehat{\beta}_{FGLasso}$}
	
	In this section, we discuss the consistency and asymptotic properties of $\widehat{\beta}_{FGLasso}$ defined in (\ref{e.fglasso}). As the main result, we show that $\widehat{\beta}_{FGLasso}$ and $\widehat{\beta}_{GLS}$ are asymptotically equivalent. For this, we assume the following assumption.
		
	\begin{assumption}[\bf Regularity Condition] ~\ \label{assumption 3} We assume the following
		\begin{itemize}
			\item[(i)] The precision matrix $\Omega=\Sigma^{-1}$ satisfies $s_{\min}(\Omega) > 0$ and $\vertii{\Omega}_\infty = \mathcal{O}(1)$.
			\item[(ii)] The regressors $X_i \in \mathbb{R}^{T\times K}$ $(i=1,2,\cdots,N)$ in model (\ref{m.stackoverT}) satisfy  $\frac{1}{T}X_i'X_j \xrightarrow{p} W_{ij} \in \mathbb{R}^{K \times K}$ and $\underset{1\leq i \leq N}{\mathrm{sup}}\vertii{\frac{1}{T}X_iX_i'-W_{ii}}_\infty=o_p(1)$ as $N,T \to \infty$.
			\item[(iii)] The singular value of $W_{ii}$ is lower bounded, i.e, $\underset{1\leq i \leq N}{\mathrm{min}}s_{min}(W_{ii})\geq c_0.$
			\item[(iv)] There exists $C_N$ such that $\underset{i,j=1,2,\hdots, N}{\max}\vertii{W_{ij}}_\infty \leq C_N$.
		\end{itemize}
	\end{assumption}
	
	The first assumption in Assumption \ref{assumption 3} regulates the true precision matrix $\Omega$. In this paper, we assume  its smallest eigenvalue does not shrink to zero and the largest elements does not explode as $N,T \to \infty$. Conditions (ii) and (iii) ensure all individual OLS estimators, $\widehat{\beta}_{i,OLS}, i=1,...,N$ are well defined.
		
	The uniform upper bound $C_N$ in condition (iv) controls the second moment of $X_t$. In this paper, we do not restrict $C_N$ to be bounded, and allow $C_N \to \infty$ slowly as $N \to \infty$. As we will show later, $C_N$ is one of the key parameters that determine the performance of graphical lasso estimator. 
	
	Some examples of the data generating processes of $\{X_{it}\}$ where $C_N=\mathcal{O}_p(1)$ are as follows:
	\begin{itemize}
		\item[(i)] If $X$ has bounded support, say $h\leq X_{it}\leq H$ for universal constants $h,H$, than $C_N=\mathcal{O}(1)$;
		\item[(ii)] If $\widetilde{X}_t=(X_{1t}',\cdots,X_{Nt}')\in \mathbb{R}^{NK}$ are random vectors with mean zero \footnote{The zero mean assumption is only for simplicity. To see this, model (\ref{m.system}) still holds for $Y_t-E(Y_t)$, $X_t-E(X_t)$.}  with covariance $W$ such that each $X_t/\sqrt{W}$ is iid with a sub-Gaussian distribution. Then $C_N=\mathcal{O}_p(1)$ if $T>cN$ for some constant $c$ and $\vertii{W}_{op}=\mathcal{O}(1)$ (see details in \cite{vershynin2018high} Theorem 4.6.1);
		\item[(iii)] If $\widetilde{X}_t \sim N(0, W)$ is zero-mean random vector generated from auto-regressive process with $s_{max}(\Theta)\leq \gamma<1$:
		\begin{equation}
		\widetilde{X}_t=\Theta \widetilde{X}_{t-1}+V_t.
		\end{equation} 
		Then $C_N=\mathcal{O}_p(1)$ if $T>c'N$ for some constant $c'>0$ and $\|W \|_{op}=\mathcal{O}(1)$. (details in Lemma 4 in \citep{negahban2011estimation}).
	\end{itemize}
	
	For the rest of the paper, we focus on the general situation thus do not put any constrain on $C_N$, i.e, we allow $C_N\to \infty$ as $N\to \infty$.

	\begin{proposition}[\textbf{Uniform Convergence Rate}] \label{prop.uniform.rate}
		Assume Assumptions \ref{assumption1}, \ref{p.assumption2} and \ref{assumption 3} hold. If $T \geq c D_N^2 \log N$ for some constant $c>0$, $\widehat{\beta}_{Fglasso}$ satisfies:
		\begin{equation}
		\vertii{\widehat{\beta}_{FGLasso}-\widehat{\beta}_{GLS}}_\infty 
		\leq \mathcal{O}_p \left(\frac{C_N^{\frac{3}{2}} D_N^{2}\sqrt{N}\log N}{\sqrt{T}} \right)
		\end{equation}
	\end{proposition}
	
	The sparsity parameter $D_N$, nonzero entries per row in $\Omega$, determines the minimum requirement of sample size $T$ and the convergence rate. As shown in the proof, the accuracy of $\widehat{\beta}_{Fglasso}$ in terms of $\widehat{\beta}_{GLS}$ depends heavily on the row norm of $\widehat{\Omega}_{gl}-\Omega$, which is bounded by $\mathcal{O}_p \left( D_N \sqrt{\frac{\log N}{T}} \right)$. 
	
	If $C_N$ and $D_N$ are bounded, then $\vertii{\widehat{\beta}_{FGLasso}-\widehat{\beta}_{GLS}}_\infty \leq \mathcal{O}_p \left( \frac{\sqrt{N}\log N}{\sqrt{T}} \right)$. As long as $\sqrt{T}$ grows faster than $\sqrt{N}\log N$,  asymptotically, $\widehat{\beta}_{FGLasso}$ is going to perform similarly to $\widehat{\beta}_{GLS}$.  
	
	Instead of $D_N=\mathcal{O}_p(1)$, if we assume $D_N$ grows as $N$ but in a much slower rate, say $D_N=\mathcal{O}_p(\log N)$ as suggested by \cite{barabasi2016network} in network literature, then $\| \widehat{\beta}_{FGLasso}-\widehat{\beta}_{GLS} \|_\infty
	= \mathcal{O}_p \left(\frac{\sqrt{N}(\log N)^{3}}{\sqrt{T}}\right)$.
	
	Next we will discuss the asymptotic property between $\widehat{\beta}_{Fglasso}$ and true $\beta$. Before we present the main result, Theorem \ref{thm.asy.fglasso}, we first show an asymptotic property of $\widehat{\beta}_{GLS}$.
	
%
	
	\begin{proposition} [\textbf{Asymptotic Property of GLS Estimator}]  Assume Assumptions \ref{assumption1} and \ref{assumption 3} hold. Then, for any $b \in \mathbb{R}^{KN\times 1}$ such that $b'b=1$, $\widehat{\beta}_{GLS}$ satisfies:
		\begin{equation}
		b'\sqrt{T}(\widehat{\beta}_{GLS}-\beta) \Rightarrow N\Big(0,b'\left[ E ( X_t'\Omega X_t ) \right]^{-1} b\Big)
		\end{equation} 
		\label{prop.asy.gls}
	\end{proposition}

	In particular, let $b=e_i \in \mathbb{R}^{KN}$, the column vector that only $i'th$ element is $1$ and $0$ otherwise. Then for each element of $\widehat{\beta}_{GLS,i}$ $(i=1,2,\cdots,KN)$, we have:
		\begin{equation}
		\sqrt{T}(\widehat{\beta}_{GLS,i}-\beta_i) \Rightarrow N \left(0, \left[ E (X_t'\Omega X_t )^{-1} \right]_{ii} \right).
		\end{equation}
	
	Combining the results from Propositions \ref{prop.uniform.rate} and \ref{prop.asy.gls}, we deduce that if $T$ grows fast enough compared with $(D_N,N)$, then $\widehat{\beta}_{Fglasso}$ is asymptotically equivalent with $\widehat{\beta}_{GLS}$, therefore the distribution of $\widehat{\beta}_{Fglasso}$ tends to a normal distribution asymptotically. 
	Summerizing this, we provide the following theorem as the main theoretical result of the paper.
	
	\begin{theorem}[\textbf{Asymptotic Property of $\widehat{\beta}_{Fglasso}$}] \label{thm.asy.fglasso}
		Assume Assumptions \ref{assumption1}, \ref{p.assumption2} and \ref{assumption 3} hold. If $T$ grows faster at $N$ such that $\frac{C_N^{3/2} D_N^2 \sqrt{N} \log N}{\sqrt{T}}\to 0$, for $T > c_0D_N^2\log N$ ($c_0$ constant) and any vector $b \in \mathbb{R}^{KN\times 1}$ such that $b'b=1$, the feasible graphical lasso estimator $\widehat{\beta}_{Fglasso}$ satisfies:
		\begin{equation}
		b'\sqrt{T}(\widehat{\beta}_{FGLasso}-\beta) \Rightarrow N\left(0,b'\left[ E\big( X_t'\Omega X_t\big)\right] ^{-1}  b \right).
		\end{equation}
	\end{theorem}
	Similar to Proposition \ref{prop.asy.gls}, let $b=e_i \in \mathbb{R}^{KN}$, Theorem \ref{thm.asy.fglasso} implies that each element of $\widehat{\beta}_{FGLasso,i}$ $(i=1,2,\cdots,KN)$ satisfies:
		\begin{equation}
		\sqrt{T}(\widehat{\beta}_{Fglasso,i}-\beta_i) \Rightarrow N \left( 0,  \left[ E\big( X_t'\Omega X_t\big)^{-1} \right]_{ii} \right).
		\end{equation}

\section{Monte Carlo Simulations}	

	In this section, we discuss finite sample properties of the FGLasso estimator using Monte Carlo simulation experiments. For different pairs of $(N,T)$, we generate data and calculate FGLasso as well as OLS, GLS and FGLS estimators. Also, we compare their distances to the true coefficients $\beta$ based on the element-wise maximum norm and root mean square error (RMSE). 
	
	\subsection{MC Design}
	\begin{figure} 
		\centering
		\includegraphics[scale=0.7]{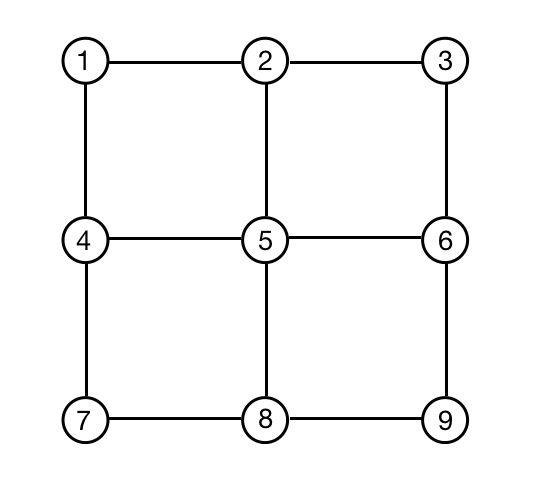}
		\caption{Four-Nearest Neighbor Lattices (N=3)} 
		\floatfoot{\textit{Notes: In figure \ref{fig:fig1}, there are $9$ nodes, each nodes have different degrees. For example, node $1$ is linked with node $4$ and $2$, thus node $1$'s degree is 2 and in the first row of $\Omega$, only $\Omega_{11}, \Omega_{12}, \Omega_{14}\neq 0$. The main purpose we consider generate $\Omega$ in this structure is the highest possible degree $d$ is bounded by $4$. To generate $\Omega$, if node $i$ is not linked with node $j$ ($i \neq j$), set $\Omega_{ij}=0$, otherwise, let $\Omega_{ij}=0.25$ and $\Omega_{ii}=1$.}}
		\label{fig:fig1}
	\end{figure}
	
	The DGP is  
		\[
			Y_t=X_t'\beta+U_t \quad t=1,2,....,T,
		\]
	where $K=1$, $X_{it} \sim N(0,1)$ for all $t$ and $i$, $\beta_{L\times 1} \sim U[-1,1]$. Further, let $U_{t} \sim N(0, \Omega^{-1})$ where $\Omega \in \mathbb{R}^{N\times N}$ is generated from the following four different designs:
	
	\begin{itemize}
	\item[(i)] {\bf Band Graph:} Let $\Omega_{i,i}=1$, $\Omega_{i,i+1}=\Omega_{i+1,i}=0.6$, $\Omega_{i,i+2}=\Omega_{i+2,i}=0.3$, and $\Omega_{i,j}=0$ for $| i-j | \geq 3$;
	\item[(ii)] {\bf Four-Nearest Neighbor Lattices Graph:} This design comes from \cite{ravikumar2011high} (see Figure \ref{fig:fig1}). Let $\Omega_{ii}=1$, $\Omega_{ij}=0.25$ if $(i,j) \in E(\Omega)$ and $0$ otherwise. For example, if $N=9$ ($3\times 3$ graph), 
	\[{\Large \Omega} \,= \;
	\begin{blockarray}{cccccccccc}
	1 & 2 & 3 & 4 & 5 & 6 & 7 & 8 & 9 & \text{Nodes} \\
	\begin{block}{(ccc|ccc|ccc)c}
	1 & 0.25 & 0   &   0.25 & 0 & 0   &0  &0  &0 & 1  \\
	0.25 & 1 & 0.25   &0  & 0.25 & 0   & 0 & 0 & 0 & 2\\
	0 & 0.25 & 1 & 0 & 0 & 0.25 & 0 & 0 & 0 & 3\\ 
	\BAhhline{---------}
	0.25 & 0 & 0 &1 & 0.25 & 0 & 0.25 & 0 & 0 & 4\\
	0 & 0.25 & 0 & 0.25 & 1 & 0.25 & 0 & 0.25 & 0 & 5\\
	0 & 0 & 0.25 & 0 & 0.25 & 1 & 0 & 0 & 0.25 & 6 \\ 
	\BAhhline{---------}
	0 & 0 & 0 & 0.25 & 0 & 0 & 1 & 0.25 & 0 & 7\\
	0 & 0 & 0 & 0 & 0.25 & 0 & 0.25 & 1 & 0.25 & 8\\
	0 & 0 & 0 & 0 & 0 & 0.25 & 0 & 0.25 & 1 & 9 \\
	\end{block}
	\end{blockarray}
	\]
	\item[(iii)] {\bf AR(1):} Let $\Omega_{ij}=0.6^{|i-j|}$;
	\item[(iv)] {\bf Dense:}  Let the covariance matrix $\Sigma=\Omega^{-1}$ be the band matrix where $\Sigma_{ii}=1$, $\Sigma_{i,i+1}=\Sigma_{i+1,i}=0.2$ and $\Sigma_{ij}=0$ for all $|i-j|\geq 2$.
	\end{itemize}
	
	The first two designs generate a sparse precision matrix $\Omega$ with a certain pattern. Specifically, the number of nonzero entries per row is always $3$ in band structure, and is at most $4$ in four-nearest neighbor lattice. 
	
	The precision matrix $\Omega$ from the AR(1) design can be seen as a special case of band graph as design 1, the values of the entries exponentially decay as they move away from the diagonal. When $N$ is small, $\Omega$ is relatively dense but becomes a sparse matrix when $N$ is large. 
	
	The precision matrix $\Omega$ from the last design is dense. Instead of precision matrix to be sparse, we consider the case that covariance matrix $\Sigma$ is sparse and has a banded structure. 
	
	We estimate $\widehat{\Omega}_{gl}$ by solving (\ref{e.gl}) using algorithm proposed by \cite{friedman2008sparse}. We choose $\lambda_n$ by the 5-fold cross validation.\footnote{More precisely, in each replication, we divide the $T$ samples into 5 folds and use four of them as the training data set and one as the validation set. With each choice of $\lambda_n$, we estimate the $\widehat{\beta}_{Fglasso}$ estimators using the training data, then plug them into the validation set and calculate the mean squared error. We choose $\lambda_n$ that minimizes the averaged MSE.} 

	\subsection{MC Results}
	For each experiment, we fix $T=200$ and let $N \in \{ 50,100,200,300,400 \}$\footnote{For four-nearest neighbor lattices design, $N$ can only be square number, so we choose $N=\{49,100,196,289,400\}$. }. For each pair of $\{N,T\}$, we compare $\| \widehat{\beta}_{OLS}-\beta\|$, $\| \widehat{\beta}_{GLS} - \beta\|$, $\|\widehat{\beta}_{FGLS}-\beta\|$ and $\| \widehat{\beta}_{FGLasso}-\beta \|$ in terms of the element-wise maximum $(l_\infty)$ norm and the root mean square error ($RMSE$)\footnote{Here, the $l_\infty$ norm is the element-wise maximum norm $\|\cdot \|_\infty$, and the RMSE is defined as $\| \cdot \|_F / \sqrt{KN}$.}. Table \ref{Tab1} reports the results$(\times 100)$ on the average of 100 replications, as well as the number of times out of 100 that $\widehat{\beta}_{FGLasso}$ outperforms $\widehat{\beta}_{FGLS}$.
	
	First, the results in Table \ref{Tab1} confirms that GLS, as the infeasible efficient estimator, performs better than other estimators and has a lower standard deviation. It also confirms the well-known results that when $N\leq T$, FGLS estimator exists and performs closer to GLS when $N$ is small relatively but gets worse as $N$ rises. In the dense design, when the efficiency gain of GLS is little, FGLS behaves worse than OLS. 
	
	For the first two designs where $\Omega$ is exactly sparse, the FGLasso estimator outperforms the FGLS estimator even when $N$ is relatively small. Moreover, it maintains the good performance when $N$ increases even when $N>T$. For example, in the band structure where  $N=300$ and $T=200$, the $l_\infty$ norm of the $(\widehat{\beta}_{FGLasso}-\beta)$ is $0.2656$ with standard deviation $0.035$, very close to the GLS estimator, which are $0.2242$ and $0.0314$, respectively. 
	
	When $N$ large, the FGLasso estimator in the third AR(1) design behaves similarly to the FGLasso in the band design. It is because the entries per row decay exponentially as it moves away from the diagonal, so $\Omega$ has a sparseness structure when $N$ is large. When $N$ is relatively small, both the FGLS and FGLasso estimators perform well and are close to the GLS estimator, but there is no significant evidence of the advantage of the FGLasso over the FGLS in finite samples. For example when $N=50$ and $T=200$, among 100 simulations, the number of times the FGLasso estimator beats the FGLS estimator is 44 times in $l_\infty$ norm, and 51 times in the RMSE.
	
	In the fourth design, when the covariance matrix has a band structure, the efficiency gain of the GLS estimator is limited. Though the FGLasso estimator performs better than the FGLS estimator, their performances are not significantly better compared to the OLS estimator.
	
	In general, when $\Omega$ is sparse, the FGLasso estimator outperforms the FGLS estimator and behaves closely to the infeasible GLS estimator even when $N>T$. When $\Omega$ is not sparse, it requires larger $T$ (or less $N$) for the FGLS estimator to perform well, as the theory predicted.  
	
\begin{table}[htbp]
\centering
\begin{threeparttable}
\scalebox{0.9}{
\begin{tabular}{cccccclccccc}
\hline\hline
T=200 & \multicolumn{5}{c}{$l_\infty \times 100$} &  & \multicolumn{5}{c}{$RMSE \times 100$} \\ \cline{2-6} \cline{8-12} 
N & 50 & 100 & 200 & 300 & 400 &  & 50 & 100 & 200 & 300 & 400 \\ \hline
\multicolumn{1}{l}{} & \multicolumn{1}{l}{} & \multicolumn{1}{l}{} & \multicolumn{1}{l}{} & \multicolumn{1}{l}{} & \multicolumn{1}{l}{} &  & \multicolumn{1}{l}{} & \multicolumn{1}{l}{} & \multicolumn{1}{l}{} & \multicolumn{1}{l}{} & \multicolumn{1}{l}{} \\
 & \multicolumn{11}{c}{Band} \\
OLS & 30.77 & 34.12 & 37.43 & 38.53 & 39.94 &  & 12.13 & 12.19 & 12.22 & 12.31 & 12.39 \\
\multicolumn{1}{l}{} & \multicolumn{1}{l}{\textit{(5.11)}} & \multicolumn{1}{l}{\textit{(4.74)}} & \multicolumn{1}{l}{\textit{(4.59)}} & \multicolumn{1}{l}{\textit{(4.7)}} & \multicolumn{1}{l}{\textit{(4.57)}} &  & \textit{(1.17)} & \textit{(0.84)} & \textit{(0.53)} & (0.49) & (0.46) \\
GLS & 17.85 & 19.57 & 20.69 & 22.42 & 22.78 &  & 7.05 & 7.11 & 7.05 & 7.15 & 7.16 \\
\multicolumn{1}{l}{} & \multicolumn{1}{l}{\textit{(3.37)}} & \multicolumn{1}{l}{\textit{(2.91)}} & \multicolumn{1}{l}{\textit{(2.39)}} & \multicolumn{1}{l}{\textit{(3.14)}} & \multicolumn{1}{l}{\textit{(2.85)}} &  & \textit{(0.81)} & \textit{(0.54)} & \textit{(0.37)} & \textit{(0.29)} & \textit{(0.25)} \\
FGLS & 20.7 & 26.10 & 37.34 &  &  &  & 8.11 & 9.34 & 12.17 &  &  \\
\multicolumn{1}{l}{} & \multicolumn{1}{l}{\textit{(3.3)}} & \multicolumn{1}{l}{\textit{(3.58)}} & \multicolumn{1}{l}{\textit{(4.55)}} & \multicolumn{1}{l}{\textit{}} & \multicolumn{1}{l}{\textit{}} &  & \textit{(0.84)} & \textit{(0.70)} & \textit{(0.53)} & \textit{} & \textit{} \\
FGLasso & 19.37 & 21.82 & 24 & 26.56 & 27.63 &  & 7.59 & 7.91 & 8.11 & 8.45 & 8.58 \\
\multicolumn{1}{l}{} & \multicolumn{1}{l}{\textit{(3.19)}} & \multicolumn{1}{l}{\textit{(3.23)}} & \multicolumn{1}{l}{\textit{(2.77)}} & \multicolumn{1}{l}{\textit{(3.5)}} & \multicolumn{1}{l}{\textit{(2.97)}} &  & \textit{(0.85)} & \textit{(0.63)} & \textit{(0.39)} & \textit{(0.33)} & \textit{(0.34)} \\
Percentage & 71 & 95 & 100 &  &  &  & 96 & 100 & 100 &  &  \\ \hline
\multicolumn{1}{l}{} & \multicolumn{1}{l}{} & \multicolumn{1}{l}{} & \multicolumn{1}{l}{} & \multicolumn{1}{l}{} & \multicolumn{1}{l}{} &  & \multicolumn{1}{l}{} & \multicolumn{1}{l}{} & \multicolumn{1}{l}{} & \multicolumn{1}{l}{} & \multicolumn{1}{l}{} \\
 & \multicolumn{11}{c}{Four-Nearest Neighbor Lattice} \\
OLS & 22.35 & 26.17 & 29.62 & 31.05 & 34.50 &  & 8.84 & 9.36 & 9.66 & 9.92 & 10.35 \\
\multicolumn{1}{l}{} & \multicolumn{1}{l}{\textit{(3.95)}} & \multicolumn{1}{l}{\textit{(3.92)}} & \multicolumn{1}{l}{\textit{(3.84)}} & \multicolumn{1}{l}{\textit{(3.97)}} & \multicolumn{1}{l}{\textit{(4.01)}} &  & \textit{(1.07)} & \textit{(0.68)} & \textit{(0.51)} & \textit{(0.44)} & \textit{(0.50)} \\
GLS & 17.54 & 19.59 & 20.62 & 21.74 & 22.67 &  & 7.01 & 7.11 & 7.07 & 7.06 & 7.15 \\
\multicolumn{1}{l}{} & \multicolumn{1}{l}{\textit{(3.17)}} & \multicolumn{1}{l}{\textit{(2.76)}} & \multicolumn{1}{l}{\textit{(2.53)}} & \multicolumn{1}{l}{\textit{(2.56)}} & \multicolumn{1}{l}{\textit{(2.63)}} &  & \textit{(0.68)} & \textit{(0.51)} & \textit{(0.43)} & \textit{(0.28)} & \textit{(0.29)} \\
FGLS & 19.31 & 23.39 & 29.53 &  &  &  & 7.73 & 8.48 & 9.60 &  &  \\
\multicolumn{1}{l}{} & \multicolumn{1}{l}{\textit{(2.90)}} & \multicolumn{1}{l}{\textit{(3.54)}} & \multicolumn{1}{l}{\textit{(4.01)}} & \multicolumn{1}{l}{\textit{}} & \multicolumn{1}{l}{\textit{}} &  & \textit{(0.71)} & \textit{(0.65)} & \textit{(0.50)} & \textit{} & \textit{} \\
FGLasso & 18.12 & 20.31 & 21.36 & 22.74 & 23.80 &  & 7.28 & 7.4 & 7.36 & 7.38 & 7.79 \\
\multicolumn{1}{l}{} & \multicolumn{1}{l}{\textit{(3.04)}} & \multicolumn{1}{l}{\textit{(3.16)}} & \multicolumn{1}{l}{\textit{(2.50)}} & \multicolumn{1}{l}{\textit{(2.66)}} & \multicolumn{1}{l}{\textit{(2.95)}} &  & \textit{(0.73)} & \textit{(0.57)} & \textit{(0.46)} & \textit{(0.30)} & \textit{(0.31)} \\
Percentage & 70 & 90 & 100 &  &  &  & 93 & 100 & 100 &  &  \\ \hline
\multicolumn{1}{l}{} & \multicolumn{1}{l}{} & \multicolumn{1}{l}{} & \multicolumn{1}{l}{} & \multicolumn{1}{l}{} & \multicolumn{1}{l}{} &  & \multicolumn{1}{l}{} & \multicolumn{1}{l}{} & \multicolumn{1}{l}{} & \multicolumn{1}{l}{} & \multicolumn{1}{l}{} \\
 & \multicolumn{11}{c}{AR(1)} \\
OLS & 27.05 & 28.63 & 30.84 & 32.27 & 33.06 &  & 10.34 & 10.27 & 10.25 & 10.36 & 10.37 \\
\multicolumn{1}{l}{} & \multicolumn{1}{l}{\textit{(5.11)}} & \multicolumn{1}{l}{\textit{(4.14)}} & \multicolumn{1}{l}{\textit{(4.04)}} & \multicolumn{1}{l}{\textit{(3.90)}} & \multicolumn{1}{l}{\textit{(4)}} &  & \textit{(0.96)} & \textit{(0.82)} & \textit{(0.50)} & \textit{(0.43)} & \textit{(0.38)} \\
GLS & 18.01 & 19.62 & 20.83 & 22.22 & 22.92 &  & 7.17 & 7.10 & 7.06 & 7.16 & 7.14 \\
\multicolumn{1}{l}{} & \multicolumn{1}{l}{\textit{(3.33)}} & \multicolumn{1}{l}{\textit{(2.75)}} & \multicolumn{1}{l}{\textit{(2.47)}} & \multicolumn{1}{l}{\textit{(2.88)}} & \multicolumn{1}{l}{\textit{(2.82)}} &  & \textit{(0.66)} & \textit{(0.52)} & \textit{(0.37)} & \textit{(0.29)} & \textit{(0.25)} \\
FGLS & 20.03 & 24.04 & 30.78 &  &  &  & 7.91 & 8.73 & 10.22 &  &  \\
\multicolumn{1}{l}{} & \multicolumn{1}{l}{\textit{(3.57)}} & \multicolumn{1}{l}{\textit{(3.66)}} & \multicolumn{1}{l}{\textit{(4.04)}} & \multicolumn{1}{l}{\textit{}} & \multicolumn{1}{l}{\textit{}} &  & \textit{(0.80)} & \textit{(0.66)} & \textit{(0.50)} & \textit{} & \textit{} \\
FGLasso & 20.34 & 22.73 & 24.72 & 26.95 & 28.19 &  & 7.90 & 8.21 & 8.38 & 8.60 & 8.68 \\
\multicolumn{1}{l}{} & \multicolumn{1}{l}{\textit{(3.63)}} & \multicolumn{1}{l}{\textit{(3.57)}} & \multicolumn{1}{l}{\textit{(2.79)}} & \multicolumn{1}{l}{\textit{(3.75)}} & \multicolumn{1}{l}{\textit{(2.99)}} &  & \textit{(0.79)} & \textit{(0.63)} & \textit{(0.41)} & \textit{(0.34)} & \textit{(0.31)} \\
Percentage & 44 & 74 & 97 &  &  &  & 51 & 98 & 100 &  &  \\ \hline
\multicolumn{1}{l}{} & \multicolumn{1}{l}{} & \multicolumn{1}{l}{} & \multicolumn{1}{l}{} & \multicolumn{1}{l}{} & \multicolumn{1}{l}{} &  & \multicolumn{1}{l}{} & \multicolumn{1}{l}{} & \multicolumn{1}{l}{} & \multicolumn{1}{l}{} & \multicolumn{1}{l}{} \\
 & \multicolumn{11}{c}{Dense} \\
OLS & 18.32 & 19.5 & 21.21 & 22.59 & 22.7 &  & 7.1 & 7.08 & 7.06 & 7.1 & 7.13 \\
\multicolumn{1}{l}{} & \multicolumn{1}{l}{\textit{(3.49)}} & \multicolumn{1}{l}{\textit{(2.97)}} & \multicolumn{1}{l}{\textit{(2.81)}} & \multicolumn{1}{l}{\textit{(2.58)}} & \multicolumn{1}{l}{\textit{(2.6)}} &  & \textit{(0.76)} & \textit{(0.54)} & \textit{(0.40)} & \textit{(0.3)} & \textit{(0.24)} \\
GLS & 17.68 & 18.48 & 20.37 & 21.53 & 21.78 &  & 6.79 & 6.77 & 6.75 & 6.8 & 6.83 \\
\multicolumn{1}{l}{} & \multicolumn{1}{l}{\textit{(3.34)}} & \multicolumn{1}{l}{\textit{(2.72)}} & \multicolumn{1}{l}{\textit{(2.63)}} & \multicolumn{1}{l}{\textit{(2.36)}} & \multicolumn{1}{l}{\textit{(2.61)}} &  & \textit{(0.74)} & \textit{(0.51)} & \textit{(0.35)} & \textit{(0.28)} & \textit{(0.23)} \\
FGLS & 18.91 & 20.9 & 21.27 &  &  &  & 7.36 & 7.66 & 7.08 &  &  \\
\multicolumn{1}{l}{} & \multicolumn{1}{l}{\textit{(3.83)}} & \multicolumn{1}{l}{\textit{(2.98)}} & \multicolumn{1}{l}{\textit{(2.89)}} & \multicolumn{1}{l}{\textit{}} & \multicolumn{1}{l}{\textit{}} &  & \textit{(0.80)} & \textit{(0.56)} & \textit{(0.40)} & \textit{} & \textit{} \\
FGLasso & 17.88 & 19.43 & 21.10 & 22.31 & 22.47 &  & 6.96 & 7.02 & 7.03 & 7.1 & 7.12 \\
\multicolumn{1}{l}{} & \multicolumn{1}{l}{\textit{(3.40)}} & \multicolumn{1}{l}{\textit{(2.82)}} & \multicolumn{1}{l}{\textit{(2.85)}} & \multicolumn{1}{l}{\textit{(2.34)}} & \multicolumn{1}{l}{\textit{(2.56)}} &  & \textit{(0.76)} & \textit{(0.53)} & \textit{(0.38)} & \textit{(0.3)} & \textit{(0.2)} \\
Percentage & 68 & 70 & 58 &  &  &  & 90 & 100 & 69 &  &  \\ \hline\hline
\end{tabular}} 
\caption{Small Sample MC Results}
\label{Tab1}
\begin{tablenotes}[flushleft]
\item \floatfoot{Notes: This table reports $\|\widehat{\beta}_{OLS}-\beta\|$, $\|\widehat{\beta}_{GLS}-\beta\|$, $\|\widehat{\beta}_{FGLS}-\beta\|$, $\|\widehat{\beta}_{FGLasso}-\beta\|$ by $l_\infty (\times 100)$, RMSE$(\times 100)$ when $\Omega$ is generated from different designs, see details in section 4.1. The percentage means the number of times $\|\widehat{\beta}_{FGLasso}-\beta\|\leq \|\widehat{\beta}_{FGLS}-\beta\|$ out of 100 simulations. The value in the parenthesis is the corresponding standard deviation$(\times 100)$. Note when $N>T$, the FGLS estimator is not well defined. All the reported results are based on 100 simulation replications.}
\end{tablenotes}
\end{threeparttable}
\end{table}

	\section{Conclusion}
	This paper proposes a new estimator $\widehat{\beta}_{FGLasso}$ in order to deal with high dimensional SUR model.
We show that under certain conditions, as $N,T$ goes to infinity at a certain rate, our FGLasso estimator is asymptotically equivalent to the GLS estimator, thus more efficient than the OLS estimator. Further, if the nonzero entries per row $(D_N)$ in the precision matrix grows much slower than the number of the equation $N$, the FGLasso estimator works well even when the number of the equations in the system $(N)$ is greater than sample size ($T$). 
	
	The key assumption under which the FGLasso estimator performs well is the sparsity of true precision matrix $\Omega$. As our knowledge, there is no consensus in the literature on how to test the true precision matrix is sparse or not, and we admit that the applicability our estimator to any general economics data remains questionable if the sparsity condition does not hold. However, we believe this is the price to pay in order to recover the $N$ by $N$ precision matrix with a limited sample size.

	\newpage
	
	\section{Appendix: Proof}
	
	\renewcommand{\theequation}{A.\arabic{equation}}  \setcounter{equation}{0}
	\renewcommand{\thetheorem}{A.\arabic{theorem}}  \setcounter{theorem}{0}
	\renewcommand{\thelemma}{A.\arabic{lemma}}  \setcounter{lemma}{0}
	\renewcommand{\thefigure}{A.\arabic{figure}}  \setcounter{figure}{0}
	
		\subsection*{Lemmas}
	Before we start the proof of Proposition \ref{prop.uniform.rate} and \ref{prop.asy.gls}, we show the following lemmas.
	
	We first summarize some useful norm inequalities from Chapter 9 in \cite{bernstein2005matrix} and Chapter in \cite{horn1990matrix}. For matrix $A$, let $s_r(A), s_{\rm max}(A), s_{\rm min}(A)$ denote the $r^{th}$, the largest, and the smallest sigular value of matrix $A$, respectively. 
	\begin{lemma} For any matrix $A,B \in \mathbb{R}^{m\times n}$, $F\in \mathbb{R}^{n \times l}$ and column vector $b \in \mathbb{R}^{n}$, we have :
		\begin{itemize}
			\item[(i)] $\vertii{A}_\infty \leq s_{\mathrm{max}}(A) \leq \vertiii{A}_\infty;$
			\item[(ii)] $\vertiii{A}_\infty \leq \sqrt{n} s_{\mathrm{max}}(A);$
			\item[(iii)] $\vertii{A+B}_\infty \leq \vertii{A}_\infty + \vertii{B}_\infty;$
			\item[(iv)] $\| AB' \|_{op} \leq \| A \|_{op} \| B' \|_{op};$
			\item[(v)] $\vertii{Ab}_\infty \leq \vertiii{A}_\infty \vertii{b}_\infty;$
			\item[(vi)] $|s_i(A)-s_i(B)|\leq \vertii{A-B}_{op}$ for each $i=1,2,\cdots,\mathrm{min}\{m,n\};$
			\item[(vii)] $s_{\mathrm{min}}(AF) \geq s_{\mathrm{min}}(A) s_{\mathrm{min}}(F);$
			\item[(viii)]For positive definite square matrix $\tilde{A}\in \mathbb{R}^{m \times m}$, $\tilde{B}\in \mathbb{R}^{n \times n}$, $s_{\mathrm{min}}(\tilde{A} \otimes \tilde{B})=s_{\mathrm{min}}(\tilde{A}) s_{\mathrm{min}}(\tilde{B})$.
		\end{itemize}
		\label{lemma: useful facts}
	\end{lemma}
	
	For convenience, we define the following notations:
	\[A_{NT}:=  \frac{1}{T}\sum_{t=1}^TX_t \Omega X_t'=\frac{1}{T} X'(\Omega \otimes I_T)X, \quad B_{NT}:=\frac{1}{\sqrt{T}}\sum_{t=1}^T X_t \Omega U_t , \]
	\[\widehat{A}_{NT}:= \frac{1}{T}\sum_{t=1}^TX_t\widehat{\Omega}_{gl} X_t'=\frac{1}{T} X'(\widehat{\Omega}_{gl}\otimes I_T)X, \quad \widehat{B}_{NT}:= \frac{1}{\sqrt{T}}\sum_{t=1}^T X_t \widehat{\Omega}_{gl} U_t, \]
	where $X_t$ and $U_t$ are defined in model (\ref{m.system}), $\widehat{\Omega}_{gl}$ is the graphical LASSO estimator defined in (\ref{e.gl}), $X$ is defined in model (\ref{m.fullequation}) and $I_T$ is the $T\times T$ identity matrix.

	\begin{lemma} \label{lemma 2.5}
		Suppose that Assumption \ref{assumption 3} holds. Then, there exists a constant $c'>0$ such that $s_{\mathrm{min}}(A_{NT})>c'$ wp1 as $N,T \to \infty$.
	\end{lemma}
	
	\begin{proof}
		Since $A_{NT}=\frac{1}{T} X'(\Omega \otimes I_T)X$, where $X=diag(X_1,\cdots, X_N)$, from Lemma \ref{lemma: useful facts} $(vii),(viii)$, we have:
		\begin{equation} \label{proof:smin(A)}
		s_{min}(A_{NT}) \geq \frac{1}{T} s_{min}^2(X) s_{min}(\Omega)
		= \min_{i=1,2,\cdots,N} \left\{s_{min}\left(\frac{X_i'X_i}{T}\right)\right\}s_{min}(\Omega).
		\end{equation}
		
		Notice that by Assumption \ref{assumption 3}, $\frac{1}{T} X_i' X_i \xrightarrow{p} W_{ii}$ as $T \to \infty$ uniformly in $i$ and $s_{min}(W_{ii}) \geq c_{ii}'$ for some constants $c_{ii}'>0$. By Lemma \ref{lemma: useful facts} $(vi)$:
		\begin{align*}
			\underset{i=1,2,\cdots,N}{\mathrm{min}}\left\{s_{min}\left(\frac{1}{T}(X_i'X_i)\right)\right\} 
			&= 	\min_{i=1,...,N} s_{min}(W_{ii}) - \max_{i=1,...,N} \left| s_{min}\left(\frac{X_i'X_i}{T} \right) - s_{min}(W_{ii})  \right| \\
			&\geq \min_{i=1,...,N} c_{ii}^{\prime} + o_p(1).
		\end{align*}
		Notice by Assumption \ref{assumption 3}, we have $s_{\min}(\Omega) \geq c > 0 $. Therefore, we have the required result for the lemma that there exists a positive constant $c'$ such that
		\begin{equation}
		s_{min}(A_{NT})\geq c s_{min}(\Omega)\geq c'.
		\end{equation}
	\end{proof}
	
	Recall the definition from Assumption \ref{p.assumption2} and \ref{assumption 3} that, $D_N$ is the maximum number of nonzero entries per row in true precision matrix $\Omega$, and $C_N\geq \underset{i,j=1,2,\hdots, N}{\max} \|  W_{ij} \|_\infty$, where $\frac{1}{T}X_i' X_j \xrightarrow{p} W_{ij}$.
	
	\begin{lemma} \label{lemma 3}
		Suppose that Assumption \ref{assumption1}, \ref{p.assumption2} and \ref{assumption 3} hold. Then, 
		\begin{itemize}
			\item[(a)] $\vertiii{\widehat{A}_{NT}-A_{NT}}_\infty \leq \mathcal{O}_p \left(C_ND_N\sqrt{\frac{\log N}{T}}\right)$
			\item[(b)] $\frac{1}{s_{\mathrm{min}}(\widehat{A}_{NT})}\leq \frac{1}{c' + \mathcal{O}_p\left(C_N D_N \sqrt{ log N / T} \right)}$
		\end{itemize}
	\end{lemma}
	
	\begin{proof} 
	(a).
		Let $\Omega=[\sigma_{ij}]$, $(i,j=1,2,\cdots,N)$.
		\begin{equation}
		A_{NT}=\frac{1}{T} X'(\Omega \otimes I_T)X=\frac{1}{T} \begin{bmatrix}
		\sigma_{11}X_1'X_1 & \sigma_{12}X_1'X_2 & \dots & \sigma_{1N} X_1'X_N \\
		\sigma_{21}X_2'X_1 & \sigma_{22}X_2'X_2 & \dots & \sigma_{2N}X_2'X_N \\
		\vdots & \vdots & \ddots & \vdots \\
		\sigma_{N1}X_N'X_1 & \sigma_{N2}X_N'X_2 & \dots & \sigma_{NN}X_N'X_N 
		\end{bmatrix}
		\end{equation}
		
		Denote $\widehat{\Omega}_{gl}-\Omega=\Delta_\Omega$. Then we have 
		\begin{equation*}
		\begin{split}
		\vertiii{\widehat{A}_{NT}-A_{NT}}_\infty&=\vertiii{\frac{1}{T}X'(\Delta_\Omega \otimes I_T)X}_\infty \\
		& \leq \underset{i=1,...,N}{\mathrm{max}} \sum_{j=1}^{N} |\Delta_{\Omega,ij}| \vertiii{W_{ij}}_\infty \\
		& \leq  \underset{i,j=1,2,\cdots,N}{\mathrm{max}}\vertiii{W_{ij}}_\infty \underset{i=1,...,N}{\mathrm{max}} \sum_{j=1}^{N} |\Delta_{\Omega,ij}| \\
		& \leq \mathcal{O}(C_N)\underset{i=1,...,N}{\mathrm{max}} \sum_{j=1}^{N} |\Delta_{\Omega,ij}| \\
		& \leq \mathcal{O}\left(C_N D_N\sqrt{\frac{\log N}{T}}\right)，
		\end{split}
		\end{equation*}
		where the second last inequality follows by the definition $ C_N:= \underset{i,j=1,2,\cdots,N}{\mathrm{max}}\vertii{W_{ij}}_\infty$. The last inequality is from Lemma \ref{lemma.ravikumar} that $\vertii{\Delta_\Omega}_\infty=\mathcal{O}_p\left(\sqrt{\frac{\log N}{T}}\right)$ and number of nonzero entries per row of $\Delta_\Omega$ is at most $D_N$.
		
		\bigskip
		
		(b). 
		According to Lemma \ref{lemma: useful facts}(vi) and Part(a), we have
		\begin{equation*}
		|s_{\mathrm{min}}(\widehat{A}_{NT})-s_{\mathrm{min}}(A_{NT})|\leq \vertii{\widehat{A}_{NT}-A_{NT}}_{op}\leq \vertiii{\widehat{A}_{NT}-A_{NT}}_\infty \leq \mathcal{O}_p \left(C_N D_N \sqrt{\frac{\log N}{T}} \right).
		\end{equation*}
		Combining this with Lemma \ref{lemma 2.5}, we have
		\begin{align*}
		\frac{1}{s_{\mathrm{min}}(\widehat{A}_{NT})} 
		&\leq \frac{1}{ s_{\mathrm{min}}(A_{NT})-|s_{\mathrm{min}}(\widehat{A}_{NT})-s_{\mathrm{min}}(A_{NT})|} \\
		&\leq \frac{1}{c' + \mathcal{O}_p\left(C_N D_N \sqrt{ \frac{\log N}{T}} \right)}.
		\end{align*}
	\end{proof}
	

    Let $G\in \mathbb{R}^{KN^2}$ be a vector consisting of $\left(\frac{1}{\sqrt{T}}\sum_{t=1}^T X_{it,k}U_{jt} : 1 \leq i,j \leq N, 1 \leq k \leq K \right)$. 
	
	%
	
	\begin{lemma}[sub-Gaussian concentration] \label{lemma:sub-Gaussian vector} Let $G_k$ be the $k^{th}$ element of $G$, for $k=1,...,KN^2$.
		Suppose that Assumption \ref{assumption1}, \ref{p.assumption2} and \ref{assumption 3} hold. Then, there exist positive constants $c$ and $c'$ such that the following two hold: 
		
		\begin{itemize}
			\item[$(a)$]	$\mathbb{P}(|G_i|> \xi \,|\, X )\leq 2 \mathrm{exp}\left(-\frac{c \xi^2}{C_N} \right)$ for all $1\leq i \leq KN^2$, where $c>0$ is a constant which doesn't dependent on $i$;
			
			\item[(b)] Conditional on $X$, $\vertii{G}_\infty \leq \mathcal{O}_p(\sqrt{C_N\log N})$.
		\end{itemize}
	\end{lemma}
	
	\begin{proof}
		(a).  In the proof due to notational simplicity, we skip writing conditioning on $X$ in the probablity and expectation notation. Recall $\Sigma$ is the covariance matrix of $U_t$, under Assumption \ref{p.assumption2},  $\mathbb{E}(e^{\lambda (U_{jt}/\Sigma_{jj})})\leq e^{\lambda^2 \sigma^2}$ for all $\lambda\in \mathbb{R}$, and $\Sigma_{jj}$ is bounded for all $1 \leq i\leq N$. Let $\bar{U}_{it}=U_{it}/\Sigma_{jj}$. Then,
		\begin{equation}
		\begin{split}
		\mathbb{P} \left(\frac{1}{\sqrt{T}}\sum_t X_{it,k}U_{jt} > \xi \right) 
		&= \mathbb{P}\left(e^{\lambda \frac{1}{\sqrt{T}}\sum_t X_{it,k}\bar{U}_{kt}} > e^{\lambda \xi  /\Sigma_{jj}}\right) \\
		&\leq e^{-\lambda \xi / \Sigma_{jj}} \mathbb{E} \left( e^{\lambda \frac{1}{\sqrt{T}}\sum_t X_{it,k}\bar{U}_{jt}} \right) \text{ (by  Markov's Inequality)} \\
		&=e^{-\lambda \xi / \Sigma_{jj}} \prod_{t=1}^T \mathbb{E} \left( e^{\lambda \frac{1}{\sqrt{T}}X_{it,k}\bar{U}_{jt}} \right) \text{  (Conditional Independency) }\\
		& \leq e^{-\lambda \xi / \Sigma_{jj}} e^{\lambda^2 \sigma^2 \frac{1}{T}\sum_{t} X_{it,k}^2}  \text{  (sub-Gaussian)} \\
		&=e^{-\frac{c \xi^2}{\sigma^2 \frac{1}{T}\sum X_{it,k}^2}}  \\
		& \leq e^{-\frac{c \xi^2}{C_N}} \quad {\rm wp1},	
		\end{split}	
		 \label{G_subgaussian}
		\end{equation}	
		where the last equality is derived from minimizing $\lambda>0$ respect to $\lambda \xi \Sigma_{jj}^* + \lambda^2 \sigma^2 \frac{1}{T}\sum_t X_{it,k}^2$.		
				
		Similarly, $P\left(\frac{1}{\sqrt{T}}\sum_t X_{it}U_{jt} < -\xi \right) \leq e^{ -\frac{c \xi^2}{C_N}}$.  This implies that $G_i$ is sub-Gaussian with parameter $c/C_N$ where $c$ doesn't dependent on $N$, i.e. \footnote{$c$ in (\ref{G_subgaussian}) might be different from $c$ in (\ref{G_i tail behavior}). By sub-Gaussian property, they differ from each other by at most an absolute constant factor (for example, see Proposition 2.5.2 in \cite{vershynin2018high}).},
		\begin{equation} \label{G_i tail behavior}
		\mathbb{E}\big(e^{\lambda G_i}\big)\leq e^{cC_N \lambda^2 }.
		\end{equation}
		
		\bigskip
		
		(b). The result in (b) follows from Markov inequality if we show $\exists c'>0 $, such that 		
		
		\begin{equation}
		\mathbb{E}\left(\underset{1\leq i \leq KN^2} {\mathrm{max}}|G_i| \, \big| \, X  \right)\leq c'  \sqrt{C_N \log N}. \label{sub_G_b}
		\end{equation}
		
		Consider $\max_{1 \leq i \leq KN^2}G_i $ first, for $\lambda >0$, by Jensen's inequality we have
		\begin{align} 
			e^{\lambda \mathbb{E} \left(\max_{1 \leq i \leq KN^2}G_i \right)} 
			&\leq \mathbb{E}\left(e^{\lambda \max_{1 \leq i \leq KN^2}G_i}\right)
			= \mathbb{E} \left(\max_{1\leq i \leq KN^2} e^{\lambda G_i} \right) \nonumber \\
			&\leq \sum_{i=1}^{KN^2} \mathbb{E}(e^{\lambda G_i}) 
			= \sum_{i=1}^{KN^2} e^{cC_N \lambda^2}\leq 
				KN^2e^{c C_N \lambda^2}.	  \label{supE}  
		\end{align}
		Take log on both sides of (\ref{supE}), we obtain for any $\lambda>0$, 
		\begin{equation}
		\mathbb{E} \left(\max_{1 \leq i \leq KN^2}G_i \right)
		\leq \frac{\log K+2\log N}{\lambda}+\lambda c C_N  .
		\end{equation}
		Minimizing the RHS with respect to $\lambda > 0$, we have   
		\[ 
		\mathbb{E} \left(\max_{1 \leq i \leq KN^2}G_i \right)
		\leq 2 \sqrt{(\log K + 2 \log N) c C_N } 
		\leq c'\sqrt{C_N \log N} 
		\]
		
		Then the result in (\ref{sub_G_b}) follows since $\underset{1 \leq i \leq KN^2}{\mathrm{max}}|G_i|=\underset{1 \leq i \leq 2KN^2}{\mathrm{max}}G_i$, where $G_{KN^2+j}=-G_j$ for $j=1,\cdots,KN^2$.
	\end{proof}

%
%
%
%
	
	\begin{lemma} \label{row norm of Omega}
	If Assumption \ref{p.assumption2} holds, $\vertiii{\widehat{\Omega}_{gl}-\Omega}_\infty \leq \mathcal{O}_p(D_N \sqrt{\frac{\log N}{T}})$.
	\end{lemma}
	
	\begin{proof}
	The result follows directly from Lemma \ref{lemma.ravikumar}. Specifically, the exact recovery result of $(ii)$ shows that $\widehat{\Omega}_{gl}$ at least remains the sparse structure of $\Omega$. That is, for the pair $(i,j)$ such that $\Omega=0$, then $\widehat{\Omega}_{gl}=0$. Therefore, we can conclude that, if there are at most $D_N$ nonzero entries per row in $\Omega$, then the maximum nonzero entries per row in $\Delta_\Omega=\widehat{\Omega}_{gl}-\Omega$ is at most $D_N$. 
	
	   Moreover, from $(i)$, we know that the largest entry-wise deviation of $\widehat{\Omega}_{gl}$ from $\Omega$ is at the rate of $\mathcal{O}_p(\sqrt{\frac{\log N}{T}})$. Therefore, by definition of the row norm $\vertiii{.}_\infty$, 
	   \[\vertiii{\widehat{\Omega}_{gl}-\Omega}_\infty \leq D_N \vertii{\widehat{\Omega}_{gl}-\Omega}_\infty = \mathcal{O}_p(D_N\sqrt{\frac{\log N}{T}}).\]
	\end{proof}

		In the following proof of Lemma \ref{lemma 4} and \ref{lemma 5}, we assume that $K=1$. We skip the general case in the proof because the proof of the general case is similar to that of $K=1$, but with more complicated notations. 
		
	\begin{lemma} \label{lemma 4}
		If Assumptions \ref{assumption1}, \ref{p.assumption2} and \ref{assumption 3} hold, then $\vertii{\widehat{B}_{NT}-B_{NT}}_\infty \leq \mathcal{O}_p \left(\frac{\sqrt{C_N} D_N \log N}{\sqrt{T}}\right)$.
	\end{lemma}
	
	\begin{proof}
		Let $\Delta_{\Omega}=\widehat{\Omega}_{gl}-\Omega$, we have,
		\[\widehat{B}_{NT}-B_{NT} =\frac{1}{\sqrt{T}}\sum_{t=1}^TX_t \Delta_{\Omega} U_t \]
		
		Let $K=1$, 
		\begin{equation*}
		\begin{split}
		\widehat{B}_{NT}-B_{NT}&=\frac{1}{\sqrt{T}}\sum_{t=1}^TX_t \Delta_{\Omega} U_t \\
		&=\frac{1}{\sqrt{T}}\sum_{t=1}^T \begin{bmatrix}
		X_{1t} &             &                   &      \\
		& X_{2t} &                    &      \\
		&             &    \ddots     &      \\
		&			&					& X_{Nt}  
		\end{bmatrix}  \begin{bmatrix}
		\Delta_{\Omega,11} & \Delta_{\Omega,12} & \dots & \Delta_{\Omega,1N} \\
		\Delta_{\Omega,21} & \Delta_{\Omega,22} & \dots & \Delta_{\Omega,2N} \\
		\vdots & \vdots & \ddots & \vdots \\
		\Delta_{\Omega,N1} & \Delta_{\Omega,N2} & \dots & \Delta_{\Omega,NN} 
		\end{bmatrix}\begin{bmatrix}
		U_{1t} \\
		U_{2t} \\
		\vdots \\
		U_{Nt}
		\end{bmatrix} \\
		&=\frac{1}{\sqrt{T}} \sum_{t=1}^T\begin{bmatrix}
		X_{1t}\Delta_{\Omega,11} & X_{1t}\Delta_{\Omega,12} & \dots & X_{1t}\Delta_{\Omega,1N} \\
		X_{2t}\Delta_{\Omega,21} & X_{2t}\Delta_{\Omega,22} & \dots & X_{2t}\Delta_{\Omega,2N} \\
		\vdots & \vdots & \ddots & \vdots \\
		X_{Nt}\Delta_{\Omega,N1} & X_{Nt}\Delta_{\Omega,N2} & \dots & X_{1t}\Delta_{\Omega,NN} 
		\end{bmatrix}
		\begin{bmatrix}
		U_{1t} \\
		U_{2t} \\
		\vdots \\
		U_{Nt}
		\end{bmatrix} \\
		& = \frac{1}{\sqrt{T}}\begin{bmatrix}
		\sum_{t=1}^TX_{1t}U_{1t}\Delta_{\Omega,11}+ \sum_{t=1}^TX_{1t}U_{2t}\Delta_{\Omega,12}+\cdots+ \sum_{t=1}^TX_{1t}U_{Nt}\Delta_{\Omega,1N} \\
		\vdots \\
		\sum_{t=1}^TX_{Nt}U_{1t}\Delta_{\Omega,N1}+ \sum_{t=1}^TX_{Nt}U_{2t}\Delta_{\Omega,N2}+\cdots+ \sum_{t=1}^TX_{Nt}U_{Nt}\Delta_{\Omega,NN} 
		\end{bmatrix}.
		\end{split} 
		\end{equation*}
		
		The $i^{th}$ element of $\widehat{B}_{NT}-B_{NT} \in \mathbb{R}^{KN \times 1}$ is:
		\begin{align*}
			(\widehat{B}_{NT}-B_{NT})_i 
			&=  \sum_{j=1}^N \frac{1}{\sqrt{T}} \sum_{t=1}^T X_{it} U_{jt} \Delta_{\Omega,ij} \\
			&\leq \left( \max_{1 \leq j \leq N} \left| \frac{1}{\sqrt{T}} \sum_{t=1}^T X_{it} U_{jt} \right| \right) \left( \max_{1 \leq i \leq N} \sum_{j=1}^N | \Delta_{\Omega,ij} | \right).
		\end{align*}
		Recall the definition $G = (\frac{1}{\sqrt{T}} \sum_{t=1}^T X_{it} U_{jt} : i,j = 1,...,N)$.
		Then by definition,
		\begin{align*}
		\| \widehat{B}_{NT}-B_{NT} \|_{\infty} 
		&= \max_{1 \leq i \leq N}  | ( \widehat{B}_{NT}-B_{NT})_i | \\
		&\leq \left( \max_{1 \leq i \leq N}\max_{1 \leq j \leq N} \left| \frac{1}{\sqrt{T}} \sum_{t=1}^T X_{it} U_{jt} \right| \right)
		\left(\max_{1 \leq i \leq N} \sum_{j=1}^N | \Delta_{\Omega,ij} | \right) \\
		&= \| G \|_{\infty} \vertiii{\Delta_\Omega}_\infty \\
		&\leq \mathcal{O}_p \left( \frac{\sqrt{C_N} D_N \log N}{\sqrt{T}} \right) .
		\end{align*}

%
		
	\end{proof}
	
	\begin{lemma} \label{lemma 5}
		If Assumptions \ref{assumption1}, \ref{p.assumption2} and \ref{assumption 3} hold, then $\vertii{B_{NT}}_\infty \leq \mathcal{O}_p(D_N\sqrt{C_N\log N}).$
	\end{lemma}
	\begin{proof}
		
		Similar as proof in Lemma \ref{lemma 4}, for $i'th$ element in $B_{NT}$ $(i=1,2,...,KN)$,
		\begin{equation*}
		\begin{split}
		|B_{NT,i}| &=\frac{1}{\sqrt{T}}|e_i' \sum_{t=1}^T X_t \Omega U_t |\\
		&=\frac{1}{\sqrt{T}}\sum_{j=1}^N \sum_{t=1}^T X_{it}U_{jt} \Omega_{ij} \\
		& \leq \left (\underset{1\leq j \leq N} {\max}\left | \frac{1}{\sqrt{T}} \sum_{t=1}^T X_{it} U_{jt} \right |\right ) \left (\sum_{j=1}^N |\Omega_{ij}| \right )
		\end{split}
		\end{equation*}
		Thus,
		\begin{equation*}
		\begin{split}
		\vertii{B_{NT}}_\infty&=\underset{i=1,2,...,N}{\mathrm{max}}|B_{NT,i}| \\
		& \leq \vertiii{\Omega}_\infty \vertii{G}_\infty \\
		& \leq D_N \vertii{\Omega}_\infty \mathcal{O}_p(\sqrt{C_N\log N}) \\
		& =\mathcal{O}_p(D_N \sqrt{C_N\log N}),
		\end{split}
		\end{equation*}
		where the inequality if from the sparsity structure of $\Omega$ from Assumption \ref{p.assumption2} and $\vertii{\Omega}_\infty=\mathcal{O}(1)$ from Assumption \ref{assumption 3}. 
	\end{proof}
	
	\subsection*{Proof of Proposition \ref{prop.uniform.rate} }
	
	\begin{proof}
	
		\begin{equation}
		\begin{split}
			\sqrt{T}\vertii{(\widehat{\beta}_{Fglasso}-\widehat{\beta}_{GLS})}_\infty &= \vertii{\widehat{A}_{NT}^{-1}\widehat{B}_{NT}-A_{NT}^{-1}B_{NT}}_\infty \\
			& = \vertii{\widehat{A}_{NT}^{-1}(A_{NT}-\widehat{A}_{NT})A^{-1}_{NT}\widehat{B}_{NT} + A^{-1}_{NT}(\widehat{B}_{NT}-B_{NT})}_\infty. 
			\end{split}
			\label{6.2.1}
		\end{equation}
		
		Apply Lemma \ref{lemma: useful facts} $(iii)$, $(v)$, continue (\ref{6.2.1}), we have:
		\begin{equation}
		\begin{split}
		\sqrt{T}\vertii{(\widehat{\beta}_{Fglasso}-\widehat{\beta}_{GLS})}_\infty & \leq \vertiii{\widehat{A}_{NT}^{-1}(A_{NT}-\widehat{A}_{NT})A^{-1}_{NT}}_\infty \vertii{\widehat{B}_{NT}}_\infty + \vertiii{A^{-1}_{NT}}_\infty\vertii{(\widehat{B}_{NT}-B_{NT})}_\infty \\
		& \leq  \vertiii{\widehat{A}_{NT}^{-1}(A_{NT}-\widehat{A}_{NT})A^{-1}_{NT}}_\infty \big(\vertii{\widehat{B}_{NT}-B_{NT}}_\infty+\vertii{B_{NT}}_\infty\big) \\
		&+\vertiii{A^{-1}_{NT}}_\infty\vertii{(\widehat{B}_{NT}-B_{NT})}_\infty.
		\end{split}
		\label{6.2.2}
		\end{equation}
		
		For the first term on the RHS, we first use Lemma \ref{lemma: useful facts} $(ii)$,
		\begin{equation*}
		\vertiii{\widehat{A}_{NT}^{-1}(A_{NT}-\widehat{A}_{NT})A^{-1}_{NT}}_\infty \leq \sqrt{NK} \vertii{ \widehat{A}_{NT}^{-1}(A_{NT}-\widehat{A}_{NT})A^{-1}_{NT} }_{op},
		\end{equation*}
		then apply Lemma \ref{lemma: useful facts} $(i)$, $(iv)$ and obtain, 
		\begin{align*}
		\vertii{\widehat{A}_{NT}^{-1}(A_{NT}-\widehat{A}_{NT})A^{-1}_{NT} }_{op}& \leq \vertii{ \widehat{A}_{NT}^{-1} }_{op} \vertii{A_{NT}-\widehat{A}_{NT}}_{op} \vertii{A^{-1}_{NT}}_{op} \\
		& \leq \vertii{ \widehat{A}_{NT}^{-1} }_{op} \vertiii{A_{NT}-\widehat{A}_{NT}}_\infty \vertii{A^{-1}_{NT}}_{op}.
		\end{align*}
	
	Therefore, (\ref{6.2.2}) can be written as:
	\begin{equation}
	\begin{split}
	\sqrt{T}\vertii{(\widehat{\beta}_{Fglasso}-\widehat{\beta}_{GLS})}_\infty & \leq \sqrt{NK} \vertii{ \widehat{A}_{NT}^{-1} }_{op} \vertiii{A_{NT}-\widehat{A}_{NT}}_\infty \vertii{A^{-1}_{NT}}_{op} \big(\vertii{\widehat{B}_{NT}-B_{NT}}_\infty+\vertii{B_{NT}}_\infty\big)\\
		& + \vertiii{A^{-1}_{NT}}_\infty\vertii{(\widehat{B}_{NT}-B_{NT})}_\infty \\
	\end{split}
	\label{6.2.3}
	\end{equation}
	
	Apply results from Lemma \ref{lemma 2.5}, \ref{lemma 3}, \ref{lemma 4} and \ref{lemma 5}, continue (\ref{6.2.3}), for constant $c>0$,
	\begin{align*}
	\sqrt{T}\vertii{(\widehat{\beta}_{Fglasso}-\widehat{\beta}_{GLS})}_\infty & \leq  \mathcal{O}_p \Big (\frac{(C_N)^\frac{3}{2} D_N^{2} \sqrt{N} \log N}{\sqrt{T}} \Big ) \frac{\mathcal{O}_p(\sqrt{\log N/T}+1)}{c+\mathcal{O}_p(C_N D_N \sqrt{\log N/T})}\\
	&+\mathcal{O}_p \Big (\frac{C_N D_N \sqrt{N} \log N}{\sqrt{T}} \Big ) \\
		& \leq \mathcal{O}_p \Big (\frac{(C_N)^\frac{3}{2} D_N^{2} \sqrt{N} \log N}{\sqrt{T}} \Big ),
		\end{align*}
	 the last inequality holds because the minimum requirement of $T>cD_N^2 \log N$.
	\end{proof}
	
	\subsection*{Proof of Proposition \ref{prop.asy.gls} }
	
	\begin{proof}
		For any bounded vector $b\in \mathbb{R}^{NK\times 1}$ such that $b'b=1$, denote scalar $\widetilde{x}_t$ as:
		\[\widetilde{x}_t:=b'\big(\frac{1}{T}\sum_{t=1}^TX_t\Omega X_t'\big)^{-1}X_t\Omega U_t,\]
		where $\mathbb{E}(\widetilde{x}_t)=0$ and $V(\widetilde{x}_t)=b'\big(\frac{1}{T}\sum_{t=1}^TX_t\Omega X_t'\big)^{-1}b \leq \frac{1}{s_{min}(A_{NT})}< \infty$.
		Then,
		\begin{equation}
		b'\sqrt{T}(\widehat{\beta}_{GLS}-\beta)=b'\big(\frac{1}{T}\sum_{t=1}^TX_t\Omega X_t'\big)^{-1} (\frac{1}{\sqrt{T}}\sum_{t=1}^TX_t\Omega U_t)=\frac{1}{\sqrt{T}}\sum_{t=1}^T \widetilde{x}_t 
		\end{equation}
		Apply CLT, we have:
		\begin{equation}
		b'\sqrt{T}(\widehat{\beta}_{GLS}-\beta) \Rightarrow N\Big(0,b'E( X_t'\Omega X_t)^{-1} b\Big)
		\end{equation}
	\end{proof}

	\newpage
	
	\section{Appendix: Additional Simulation Results}

	\renewcommand{\theequation}{B.\arabic{equation}}  		\setcounter{equation}{0}
	\renewcommand{\thetheorem}{B.\arabic{theorem}}  \setcounter{theorem}{0}
	\renewcommand{\thelemma}{A.\arabic{lemma}}  \setcounter{lemma}{0}
	\renewcommand{\thefigure}{B.\arabic{figure}}  \setcounter{figure}{0}
	\renewcommand{\thetable}{B.\arabic{table}}  \setcounter{table}{0}

		In this section, we use the same four designs as described in section 4. Let $T=\{50,100,200\}$ and control $N/T$ ratio to be $\{0.1,0.5,1,1.5,2\}$. In each table, as before, we report $l_\infty \times 100$, $RMSE\times 100$, as well as the number of times the FGLasso estimator outperforms the FGLS. We also show the penalty parameter $\lambda_n$ chosen from the 5-fold cross validation.
		
	As shown in Table \ref{tab2} and \ref{tab3}, when $\Omega$ is exactly sparse, the $FGLasso$ performs better when $N$ is relatively large compared to $T$, i.e., $N/T \geq 0.5$. Even when $N$ is twice larger than $T$, the FGLasso estimator still performs closely to the GLS in both the $l_\infty$ and the RMSE standards. 
	
	Table \ref{tab4} shows the results when $\Omega$ is generated from the AR(1) design. When $N$ is small, $\Omega$ is dense, and it performs worse than the FGLS estimator. For example, when $\{T,N\}=\{50,5\}$, the number of times that the FGLasso estimator outperforms the FGLS estimator is only around 35 out of 100. However, when $N$ is relatively large, the FGLasso estimator beats the FGLS estimator almost every time and performs closely to the GLS estimator as $N,T$ increases.
	
	Under the dense design where the true precision matrix is dense and the true covariance matrix is sparse, the efficiency gain of the GLS estimator is minor. Both the FGLS and FGLasso estimators behave poorly. Notice here when $N/T$ is relatively small, the FGLasso estimator still outperforms the FGLS estimator, and in some case, it is slightly better than the OLS estimator. 
	
	In general, when the underlying true precision matrix has a certain sparse structure, the FGLasso estimator beats the FGLS estimator and becomes closer to the GLS estimator. When $\Omega$ is dense, both the FGLasso estimator and the FGLS estimator perform poorly.
	 
\begin{table}[htbp]
\centering
\begin{threeparttable}
\resizebox{\columnwidth}{!}{
\begin{tabular}{|cccccccccccccccc|}
\hline
\multicolumn{1}{|c|}{} & \multicolumn{3}{c|}{N/T=0.1} & \multicolumn{3}{c|}{N/T=0.5} & \multicolumn{3}{c|}{N/T=1} & \multicolumn{3}{c|}{N/T=1.5} & \multicolumn{3}{c|}{N/T=2} \\ \cline{2-16} 
\multicolumn{1}{|c|}{T} & 50 & 100 & \multicolumn{1}{c|}{200} & 50 & 100 & \multicolumn{1}{c|}{200} & 50 & 100 & \multicolumn{1}{c|}{200} & 50 & 100 & \multicolumn{1}{c|}{200} & 50 & 100 & 200 \\
\multicolumn{1}{|c|}{N} & 5 & 10 & \multicolumn{1}{c|}{20} & 25 & 50 & \multicolumn{1}{c|}{100} & 50 & 100 & \multicolumn{1}{c|}{200} & 75 & 150 & \multicolumn{1}{c|}{300} & 100 & 200 & 400 \\ \hline
\multicolumn{16}{|l|}{} \\
\multicolumn{16}{|c|}{Element-wise maximum norm $l_\infty \times 100$} \\
\multicolumn{1}{|c|}{OLS} & 31.92 & 29.35 & \multicolumn{1}{c|}{25.46} & 53.94 & 43.00 & \multicolumn{1}{c|}{34.12} & 63.02 & 48.79 & \multicolumn{1}{c|}{37.43} & 67.07 & 50.73 & \multicolumn{1}{c|}{38.53} & 69.42 & 52.51 & 39.94 \\
\multicolumn{1}{|c|}{} & \textit{(12.13)} & \textit{(8.28)} & \multicolumn{1}{c|}{\textit{(5.14)}} & \textit{(12.84)} & \textit{(7.17)} & \multicolumn{1}{c|}{\textit{(4.74)}} & \textit{(11.17)} & \textit{(6.57)} & \multicolumn{1}{c|}{\textit{(4.59)}} & \textit{(10.53)} & \textit{(6.38)} & \multicolumn{1}{c|}{\textit{(4.7)}} & \textit{(10.28)} & \textit{(6.79)} & \textit{(4.57)} \\
\multicolumn{1}{|c|}{GLS} & 23.52 & 19.37 & \multicolumn{1}{c|}{15.64} & 33.52 & 25.00 & \multicolumn{1}{c|}{19.57} & 37.82 & 28.71 & \multicolumn{1}{c|}{20.69} & 37.92 & 29.88 & \multicolumn{1}{c|}{22.42} & 40.80 & 30.98 & 22.78 \\
\multicolumn{1}{|c|}{} & \textit{(9.01)} & \textit{(5.17)} & \multicolumn{1}{c|}{\textit{(2.98)}} & \textit{(7.91)} & \textit{(4.76)} & \multicolumn{1}{c|}{\textit{(2.91)}} & \textit{(6.74)} & \textit{(4.80)} & \multicolumn{1}{c|}{\textit{(2.39)}} & \textit{(5.96)} & \textit{(4.59)} & \multicolumn{1}{c|}{\textit{(3.14)}} & \textit{(6.77)} & \textit{(3.81)} & \textit{(2.85)} \\
\multicolumn{1}{|c|}{FGLS} & 23.87 & 20.27 & \multicolumn{1}{c|}{16.29} & 43.38 & 33.99 & \multicolumn{1}{c|}{26.10} & 62.25 & 48.40 & \multicolumn{1}{c|}{37.34} &  &  & \multicolumn{1}{c|}{} &  &  &  \\
\multicolumn{1}{|c|}{} & \textit{(9.42)} & \textit{(5.14)} & \multicolumn{1}{c|}{\textit{(3.35)}} & \textit{(9.41)} & \textit{(7.13)} & \multicolumn{1}{c|}{\textit{(3.58)}} & \textit{(10.93)} & \textit{(6.56)} & \multicolumn{1}{c|}{\textit{(4.55)}} & \textit{} & \textit{} & \multicolumn{1}{c|}{\textit{}} & \textit{} & \textit{} & \textit{} \\
\multicolumn{1}{|c|}{FGLasso} & 24.37 & 20.76 & \multicolumn{1}{c|}{16.43} & 40.67 & 30.13 & \multicolumn{1}{c|}{21.82} & 48.39 & 34.89 & \multicolumn{1}{c|}{24} & 54.83 & 37.34 & \multicolumn{1}{c|}{36.56} & 57.72 & 40.14 & 27.63 \\
\multicolumn{1}{|c|}{} & (10.23) & \textit{(5.43)} & \multicolumn{1}{c|}{\textit{(3.17)}} & \textit{(8.98)} & \textit{(6.41)} & \multicolumn{1}{c|}{\textit{(3.23)}} & \textit{(9.14)} & \textit{(5.31)} & \multicolumn{1}{c|}{\textit{(2.77)}} & \textit{(9.18)} & \textit{(4.95)} & \multicolumn{1}{c|}{\textit{(3.5)}} & \textit{(9.28)} & \textit{(4.69)} & \textit{(2.97)} \\
\multicolumn{1}{|c|}{Percentage} & 52 & 47 & \multicolumn{1}{c|}{48} & 77 & 85 & \multicolumn{1}{c|}{95} & 97 & 100 & \multicolumn{1}{c|}{100} &  &  & \multicolumn{1}{c|}{} &  &  &  \\ \hline
\multicolumn{16}{|c|}{} \\
\multicolumn{16}{|c|}{$RMSE \times 100$} \\
\multicolumn{1}{|c|}{OLS} & 19.35 & 15.47 & \multicolumn{1}{c|}{11.66} & 23.25 & 17.18 & \multicolumn{1}{c|}{12.19} & 24.29 & 17.48 & \multicolumn{1}{c|}{12.22} & 24.79 & 17.39 & \multicolumn{1}{c|}{12.31} & 24.90 & 17.49 & 12.39 \\
\multicolumn{1}{|c|}{} & \textit{(7.17)} & \textit{(3.68)} & \multicolumn{1}{c|}{\textit{(1.64)}} & \textit{(3.72)} & \textit{(1.94)} & \multicolumn{1}{c|}{\textit{(0.84)}} & \textit{(2.58)} & \textit{(1.20)} & \multicolumn{1}{c|}{\textit{(0.53)}} & \textit{(2.02)} & \textit{(1.05)} & \multicolumn{1}{c|}{\textit{(0.49)}} & \textit{(1.82)} & \textit{(0.88)} & \textit{(0.46)} \\
\multicolumn{1}{|c|}{GLS} & 14.09 & 10.04 & \multicolumn{1}{c|}{7.06} & 14.23 & 10.11 & \multicolumn{1}{c|}{7.11} & 14.62 & 10.19 & \multicolumn{1}{c|}{7.05} & 14.22 & 10.08 & \multicolumn{1}{c|}{7.15} & 14.54 & 10.17 & 7.16 \\
\multicolumn{1}{|c|}{} & \textit{(4.81)} & \textit{(2.07)} & \multicolumn{1}{c|}{\textit{(1.07)}} & \textit{(2.24)} & \textit{(1.16)} & \multicolumn{1}{c|}{\textit{(0.54)}} & \textit{(1.59)} & \textit{(0.71)} & \multicolumn{1}{c|}{\textit{(0.37)}} & \textit{(1.15)} & \textit{(0.58)} & \multicolumn{1}{c|}{\textit{(0.29)}} & \textit{(1.07)} & \textit{(0.54)} & \textit{(0.25)} \\
\multicolumn{1}{|c|}{FGLS} & 14.53 & 10.56 & \multicolumn{1}{c|}{7.44} & 18.29 & 13.33 & \multicolumn{1}{c|}{9,34} & 24.00 & 17.36 & \multicolumn{1}{c|}{12.17} &  &  & \multicolumn{1}{c|}{} &  &  &  \\
\multicolumn{1}{|c|}{} & \textit{(5.14)} & \textit{(2.20)} & \multicolumn{1}{c|}{\textit{(1.11)}} & \textit{(3.01)} & \textit{(1.60)} & \multicolumn{1}{c|}{\textit{(0.70)}} & \textit{(2.56)} & \textit{(1.17)} & \multicolumn{1}{c|}{\textit{(0.53)}} & \textit{} & \textit{} & \multicolumn{1}{c|}{\textit{}} & \textit{} & \textit{} & \textit{} \\
\multicolumn{1}{|c|}{FGLasso} & 14.95 & 10.82 & \multicolumn{1}{c|}{7.44} & 17.29 & 11.87 & \multicolumn{1}{c|}{7.91} & 18.78 & 12.47 & \multicolumn{1}{c|}{8.11} & 19.60 & 12.66 & \multicolumn{1}{c|}{8.45} & 20.16 & 13.02 & 8.58 \\
\multicolumn{1}{|c|}{} & \textit{(5.78)} & \textit{(2.36)} & \multicolumn{1}{c|}{\textit{(1.11)}} & \textit{(3.01)} & \textit{(1.49)} & \multicolumn{1}{c|}{\textit{(0.63)}} & \textit{(2.02)} & \textit{(0.83)} & \multicolumn{1}{c|}{\textit{(0.39)}} & \textit{(1.70)} & \textit{(0.81)} & \multicolumn{1}{c|}{\textit{(0.33)}} & \textit{(1.41)} & \textit{(0.70)} & \textit{(0.34)} \\
\multicolumn{1}{|c|}{Percentage} & 49 & 46 & \multicolumn{1}{c|}{46} & 88 & 97 & \multicolumn{1}{c|}{100} & 100 & 100 & \multicolumn{1}{c|}{100} &  &  & \multicolumn{1}{c|}{} &  &  &  \\ \hline
\multicolumn{1}{|l}{} & \multicolumn{1}{l}{} & \multicolumn{1}{l}{} & \multicolumn{1}{l}{} & \multicolumn{1}{l}{} & \multicolumn{1}{l}{} & \multicolumn{1}{l}{} & \multicolumn{1}{l}{} & \multicolumn{1}{l}{} & \multicolumn{1}{l}{} & \multicolumn{1}{l}{} & \multicolumn{1}{l}{} & \multicolumn{1}{l}{} & \multicolumn{1}{l}{} & \multicolumn{1}{l}{} & \multicolumn{1}{l|}{} \\
\multicolumn{16}{|c|}{Cross Validation} \\
\multicolumn{1}{|c|}{$\lambda_n \times 100$} & 7.01 & 6.29 & \multicolumn{1}{c|}{5.90} & 9.46 & 8.96 & \multicolumn{1}{c|}{8.27} & 15.01 & 15.09 & \multicolumn{1}{c|}{13.32} & 20.48 & 18.79 & \multicolumn{1}{c|}{15.60} & 24.22 & 22.18 & 18.79 \\
\multicolumn{1}{|c|}{} & \textit{(12.51)} & \textit{(10.22)} & \multicolumn{1}{c|}{\textit{(9.01)}} & \textit{(6.34)} & \textit{(3.36)} & \multicolumn{1}{c|}{\textit{(3.05)}} & \textit{(5.78)} & \textit{(4.44)} & \multicolumn{1}{c|}{\textit{(2.17)}} & \textit{(5.76)} & \textit{(3.47)} & \multicolumn{1}{c|}{\textit{(3.18)}} & \textit{(5.38)} & \textit{(4.53)} & \textit{(2.87)} \\ \hline 
\end{tabular}}
\caption{Band Graph}
\label{tab2}
\begin{tablenotes}[flushleft]
\item \floatfoot{Notes: This table reports $\|\widehat{\beta}_{OLS}-\beta\|$, $\|\widehat{\beta}_{GLS}-\beta\|$, $\|\widehat{\beta}_{FGLS}-\beta\|$, $\|\widehat{\beta}_{FGLasso}-\beta\|$ by $l_\infty (\times 100)$, RMSE$(\times 100)$ when $\Omega$ is generated from the band structure, see details in section 4.1. The percentage means the number of times $\|\widehat{\beta}_{FGLasso}-\beta\|\leq \|\widehat{\beta}_{FGLS}-\beta\|$ out of 100 simulations. $\lambda_n$ is the penalty parameter in (\ref{e.gl}) chosen by the 5-fold cross validation. The value in the parenthesis is the corresponding standard deviation$(\times 100)$. All the reported results here are based on 100 replications.}
\end{tablenotes}
\end{threeparttable}
\end{table}

\begin{table}[htbp]
\centering
\begin{threeparttable}
\resizebox{\columnwidth}{!}{
\begin{tabular}{|cccccccccccccccc|}
\hline
\multicolumn{1}{|c|}{} & \multicolumn{3}{c|}{N/T=0.1} & \multicolumn{3}{c|}{N/T=0.5} & \multicolumn{3}{c|}{N/T=1} & \multicolumn{3}{c|}{N/T=1.5} & \multicolumn{3}{c|}{N/T=2} \\ \cline{2-16} 
\multicolumn{1}{|c|}{T} & 50 & 100 & \multicolumn{1}{c|}{200} & 50 & 100 & \multicolumn{1}{c|}{200} & 50 & 100 & \multicolumn{1}{c|}{200} & 50 & 100 & \multicolumn{1}{c|}{200} & 50 & 100 & 200 \\
\multicolumn{1}{|c|}{N} & 5 & 10 & \multicolumn{1}{c|}{20} & 25 & 50 & \multicolumn{1}{c|}{100} & 50 & 100 & \multicolumn{1}{c|}{200} & 75 & 150 & \multicolumn{1}{c|}{300} & 100 & 200 & 400 \\ \hline
\multicolumn{16}{|c|}{} \\
\multicolumn{16}{|c|}{Element-wise maximum norm $l_\infty \times 100$} \\
\multicolumn{1}{|c|}{OLS} & 23.77 & 20.29 & \multicolumn{1}{c|}{17.40} & 39.65 & 32.86 & \multicolumn{1}{c|}{26.17} & 45.07 & 38.21 & \multicolumn{1}{c|}{29.62} & 52.42 & 40.60 & \multicolumn{1}{c|}{31.05} & 55.88 & 42.64 & 34.50 \\
\multicolumn{1}{|c|}{} & \textit{(10.31)} & \textit{(5.32)} & \multicolumn{1}{c|}{\textit{(4.20)}} & \textit{(8.03)} & \textit{(5.16)} & \multicolumn{1}{c|}{\textit{(3.92)}} & \textit{(8.07)} & \textit{(6.54)} & \multicolumn{1}{c|}{\textit{(3.84)}} & \textit{(9.57)} & \textit{(6.13)} & \multicolumn{1}{c|}{\textit{(3.97)}} & \textit{(9.80)} & \textit{(6.56)} & \textit{(4.01)} \\
\multicolumn{1}{|c|}{GLS} & 21.98 & 18.37 & \multicolumn{1}{c|}{14.60} & 34.12 & 25.42 & \multicolumn{1}{c|}{19.59} & 35.22 & 29.01 & \multicolumn{1}{c|}{20.62} & 40.63 & 29.64 & \multicolumn{1}{c|}{21.74} & 41.22 & 30.50 & 22.67 \\
\multicolumn{1}{|c|}{} & \textit{(9.16)} & \textit{(5.06)} & \multicolumn{1}{c|}{\textit{(3.53)}} & \textit{(7.35)} & \textit{(4.07)} & \multicolumn{1}{c|}{\textit{(2.76)}} & \textit{(5.77)} & \textit{(4.82)} & \multicolumn{1}{c|}{\textit{(2.53)}} & \textit{(7.29)} & \textit{(4.22)} & \multicolumn{1}{c|}{\textit{(2.56)}} & \textit{(6.43)} & \textit{(4.69)} & \textit{(2.63)} \\
\multicolumn{1}{|c|}{FGLS} & 23.14 & 18.71 & \multicolumn{1}{c|}{15.26} & 39.14 & 29.76 & \multicolumn{1}{c|}{23.39} & 45.36 & 38.42 & \multicolumn{1}{c|}{29.53} &  &  & \multicolumn{1}{c|}{} &  &  &  \\
\multicolumn{1}{|c|}{} & \textit{(10.09)} & \textit{(5.14)} & \multicolumn{1}{c|}{\textit{(3.61)}} & \textit{(8.36)} & \textit{(4.88)} & \multicolumn{1}{c|}{\textit{(3.54)}} & \textit{(8.06)} & \textit{(6.48)} & \multicolumn{1}{c|}{\textit{(4.01)}} & \textit{} & \textit{} & \multicolumn{1}{c|}{\textit{}} & \textit{} & \textit{} & \textit{} \\
\multicolumn{1}{|c|}{FGLasso} & 22.72 & 18.67 & \multicolumn{1}{c|}{14.93} & 37.11 & 26.94 & \multicolumn{1}{c|}{20.31} & 38.92 & 31.21 & \multicolumn{1}{c|}{21.36} & 44.73 & 32.34 & \multicolumn{1}{c|}{22.74} & 47.54 & 32.98 & 23.80 \\
\multicolumn{1}{|c|}{} & \textit{(10.23)} & \textit{(4.80)} & \multicolumn{1}{c|}{\textit{(3.62)}} & \textit{(7.28)} & \textit{(4.68)} & \multicolumn{1}{c|}{\textit{(3.16)}} & \textit{(7.10)} & \textit{(4.69)} & \multicolumn{1}{c|}{\textit{(2.50)}} & \textit{(8.36)} & \textit{(4.79)} & \multicolumn{1}{c|}{\textit{(2.66)}} & \textit{(7.15)} & \textit{(4.77)} & \textit{(2.95)} \\
\multicolumn{1}{|c|}{Percentage} & 55 & 56 & \multicolumn{1}{c|}{55} & 65 & 78 & \multicolumn{1}{c|}{90} & 88 & 90 & \multicolumn{1}{c|}{100} &  &  & \multicolumn{1}{c|}{} &  &  &  \\ \hline
 &  &  &  &  &  &  &  &  &  &  &  &  &  &  &  \\
\multicolumn{16}{|c|}{$RMSE \times 100$} \\
\multicolumn{1}{|c|}{OLS} & 15.16 & 11.01 & \multicolumn{1}{c|}{8.26} & 17.27 & 12.81 & \multicolumn{1}{c|}{9.36} & 17.88 & 13.35 & \multicolumn{1}{c|}{9.66} & 18.59 & 13.61 & \multicolumn{1}{c|}{9.92} & 19.13 & 13.89 & 10.35 \\
\multicolumn{1}{|c|}{} & \textit{(6.25)} & \textit{(2.63)} & \multicolumn{1}{c|}{\textit{(1.45)}} & \textit{(2.55)} & \textit{(1.17)} & \multicolumn{1}{c|}{\textit{(0.68)}} & \textit{(1.70)} & \textit{(0.97)} & \multicolumn{1}{c|}{\textit{(0.51)}} & \textit{(1.56)} & \textit{(0.91)} & \multicolumn{1}{c|}{\textit{(0.44)}} & \textit{(1.49)} & \textit{(0.72)} & \textit{(0.50)} \\
\multicolumn{1}{|c|}{GLS} & 14.07 & 9.84 & \multicolumn{1}{c|}{6.90} & 14.68 & 10.02 & \multicolumn{1}{c|}{7.11} & 14.23 & 10.20 & \multicolumn{1}{c|}{7.07} & 14.41 & 10.07 & \multicolumn{1}{c|}{7.06} & 14.57 & 10.07 & 7.15 \\
\multicolumn{1}{|c|}{} & \textit{(5.51)} & \textit{(2.41)} & \multicolumn{1}{c|}{\textit{(1.16)}} & \textit{(1.89)} & \textit{(0.88)} & \multicolumn{1}{c|}{\textit{(0.51)}} & \textit{(1.48)} & \textit{(0.68)} & \multicolumn{1}{c|}{\textit{(0.43)}} & \textit{(1.22)} & \textit{(0.64)} & \multicolumn{1}{c|}{\textit{(0.28)}} & \textit{(1.01)} & \textit{(0.58)} & \textit{(0.29)} \\
\multicolumn{1}{|c|}{FGLS} & 14.56 & 10.18 & \multicolumn{1}{c|}{7.19} & 16.97 & 11.88 & \multicolumn{1}{c|}{8.48} & 17.86 & 13.33 & \multicolumn{1}{c|}{9.60} &  &  & \multicolumn{1}{c|}{} &  &  &  \\
\multicolumn{1}{|c|}{} & \textit{(5.99)} & \textit{(2.53)} & \multicolumn{1}{c|}{\textit{(1.25)}} & \textit{(2.52)} & \textit{(1.12)} & \multicolumn{1}{c|}{\textit{(0.65)}} & \textit{(1.72)} & \textit{(0.97)} & \multicolumn{1}{c|}{\textit{(0.50)}} & \textit{} & \textit{} & \multicolumn{1}{c|}{\textit{}} & \textit{} & \textit{} & \textit{} \\
\multicolumn{1}{|c|}{FGLasso} & 14.38 & 10.14 & \multicolumn{1}{c|}{7.04} & 15.93 & 10.71 & \multicolumn{1}{c|}{7.40} & 15.69 & 10.98 & \multicolumn{1}{c|}{7.36} & 16.61 & 10.85 & \multicolumn{1}{c|}{7.38} & 16.49 & 10.91 & 7.79 \\
\multicolumn{1}{|c|}{} & \textit{(6.11)} & \textit{(2.42)} & \multicolumn{1}{c|}{\textit{(1.22)}} & \textit{(2.26)} & \textit{(1.02)} & \multicolumn{1}{c|}{\textit{(0.57)}} & \textit{(1.71)} & \textit{(0.71)} & \multicolumn{1}{c|}{\textit{(0.46)}} & \textit{(1.47)} & \textit{(0.70)} & \multicolumn{1}{c|}{\textit{(0.30)}} & \textit{(1.17)} & \textit{(0.62)} & \textit{(0.31)} \\
\multicolumn{1}{|c|}{Percentage} & 56 & 57 & \multicolumn{1}{c|}{68} & 80 & 95 & \multicolumn{1}{c|}{100} & 98 & 100 & \multicolumn{1}{c|}{100} &  &  & \multicolumn{1}{c|}{} &  &  &  \\ \hline
\multicolumn{1}{|l}{} & \multicolumn{1}{l}{} & \multicolumn{1}{l}{} & \multicolumn{1}{l}{} & \multicolumn{1}{l}{} & \multicolumn{1}{l}{} & \multicolumn{1}{l}{} & \multicolumn{1}{l}{} & \multicolumn{1}{l}{} & \multicolumn{1}{l}{} & \multicolumn{1}{l}{} & \multicolumn{1}{l}{} & \multicolumn{1}{l}{} & \multicolumn{1}{l}{} & \multicolumn{1}{l}{} & \multicolumn{1}{l|}{} \\
\multicolumn{16}{|c|}{Cross Validation} \\
\multicolumn{1}{|c|}{$\lambda_n \times 100$} & 9.53 & 6.11 & \multicolumn{1}{c|}{6.60} & 13.36 & 11.24 & \multicolumn{1}{c|}{11.66} & 17.82 & 15.79 & \multicolumn{1}{c|}{14.71} & 21.05 & 19.45 & \multicolumn{1}{c|}{16.56} & 23.38 & 21.31 & 19.28 \\
\multicolumn{1}{|c|}{} & \textit{(10.15)} & \textit{(7.48)} & \multicolumn{1}{c|}{\textit{(6.59)}} & \textit{(9.6)} & \textit{(4.58)} & \multicolumn{1}{c|}{\textit{(3.51)}} & \textit{(6.94)} & \textit{(3.83)} & \multicolumn{1}{c|}{\textit{(2.22)}} & \textit{(5.90)} & \textit{(2.97)} & \multicolumn{1}{c|}{\textit{(1.69)}} & \textit{(5.26)} & \textit{(2.97)} & \textit{(1.87)} \\ \hline 
\end{tabular}}
\caption{Four Nearest Neighbor Lattice}
\begin{tablenotes}[flushleft]
\item \floatfoot{Notes: This table reports $\|\widehat{\beta}_{OLS}-\beta\|$, $\|\widehat{\beta}_{GLS}-\beta\|$, $\|\widehat{\beta}_{FGLS}-\beta\|$, $\|\widehat{\beta}_{FGLasso}-\beta\|$ by $l_\infty (\times 100)$, RMSE$(\times 100)$ when $\Omega$ is generated from four nearest neighbor lattice structure, see details in section 4.1. Percentage means the number of times $\|\widehat{\beta}_{FGLasso}-\beta\|\leq \|\widehat{\beta}_{FGLS}-\beta\|$ out of 100 simulations. $\lambda_n$ is the penalty parameter in (\ref{e.gl}) chosen by 5-fold cross validation. Values in parenthesis are the corresponding standard deviation$(\times 100)$. All results reported here are average of 100 replications.}
\label{tab3}
\end{tablenotes}
\end{threeparttable}
\end{table}

\begin{table}[htbp]
\centering
\begin{threeparttable}
\resizebox{\columnwidth}{!}{
\begin{tabular}{|cccccccccccccccc|}
\hline 
\multicolumn{1}{|c|}{} & \multicolumn{3}{c|}{N/T=0.1} & \multicolumn{3}{c|}{N/T=0.5} & \multicolumn{3}{c|}{N/T=1} & \multicolumn{3}{c|}{N/T=1.5} & \multicolumn{3}{c|}{N/T=2} \\ \cline{2-16} 
\multicolumn{1}{|c|}{T} & 50 & 100 & \multicolumn{1}{c|}{200} & 50 & 100 & \multicolumn{1}{c|}{200} & 50 & 100 & \multicolumn{1}{c|}{200} & 50 & 100 & \multicolumn{1}{c|}{200} & 50 & 100 & 200 \\
\multicolumn{1}{|c|}{N} & 5 & 10 & \multicolumn{1}{c|}{20} & 25 & 50 & \multicolumn{1}{c|}{100} & 50 & 100 & \multicolumn{1}{c|}{200} & 75 & 150 & \multicolumn{1}{c|}{300} & 100 & 200 & 400 \\ \hline
\multicolumn{16}{|c|}{} \\
\multicolumn{16}{|c|}{Element-wise maximum norm $l_\infty \times 100$} \\
\multicolumn{1}{|c|}{OLS} & 33.67 & 28.34 & \multicolumn{1}{c|}{23.37} & 47.08 & 36.72 & \multicolumn{1}{c|}{28.63} & 53.45 & 40.89 & \multicolumn{1}{c|}{30.84} & 67.07 & 42.16 & \multicolumn{1}{c|}{32.27} & 59.74 & 44.27 & 33.06 \\
\multicolumn{1}{|c|}{} & \textit{(12.17)} & \textit{(8.4)} & \multicolumn{1}{c|}{\textit{(5.24)}} & \textit{(9.66)} & \textit{(6.21)} & \multicolumn{1}{c|}{\textit{(4.14)}} & \textit{(9.34)} & \textit{(6.15)} & \multicolumn{1}{c|}{\textit{(4.04)}} & \textit{(10.53)} & \textit{(5.69)} & \multicolumn{1}{c|}{\textit{(3.90)}} & \textit{(9.80)} & \textit{(5.30)} & \textit{(4.00)} \\
\multicolumn{1}{|c|}{GLS} & 25.10 & 19.27 & \multicolumn{1}{c|}{15.65} & 33.77 & 25.38 & \multicolumn{1}{c|}{19.62} & 37.44 & 28.73 & \multicolumn{1}{c|}{20.83} & 37.92 & 29.95 & \multicolumn{1}{c|}{22.22} & 40.52 & 31.42 & 22.92 \\
\multicolumn{1}{|c|}{} & \textit{(8.23)} & \textit{(4.73)} & \multicolumn{1}{c|}{\textit{(3.62)}} & \textit{(7.04)} & \textit{(4.74)} & \multicolumn{1}{c|}{\textit{(2.75)}} & \textit{(7.35)} & \textit{(4.50)} & \multicolumn{1}{c|}{\textit{(2.47)}} & \textit{(5.96)} & \textit{(4.54)} & \multicolumn{1}{c|}{\textit{(2.88)}} & \textit{(6.43)} & \textit{(3.81)} & \textit{(2.82)} \\
\multicolumn{1}{|c|}{FGLS} & 26.09 & 20.49 & \multicolumn{1}{c|}{16.78} & 41.55 & 32.54 & \multicolumn{1}{c|}{24.04} & 53.29 & 40.76 & \multicolumn{1}{c|}{30.78} &  &  & \multicolumn{1}{c|}{} &  &  &  \\
\multicolumn{1}{|c|}{} & \textit{(9.83)} & \textit{(5.68)} & \multicolumn{1}{c|}{\textit{(4.01)}} & \textit{(8.58)} & \textit{(6.61)} & \multicolumn{1}{c|}{\textit{(3.66)}} & \textit{(9.33)} & \textit{(6.10)} & \multicolumn{1}{c|}{\textit{(4.04)}} & \textit{} & \textit{} & \multicolumn{1}{c|}{\textit{}} & \textit{} & \textit{} & \textit{} \\
\multicolumn{1}{|c|}{FGLasso} & 27.11 & 21.02 & \multicolumn{1}{c|}{17.29} & 39.89 & 31.17 & \multicolumn{1}{c|}{22.73} & 47.35 & 35.25 & \multicolumn{1}{c|}{24.72} & 54.83 & 36.44 & \multicolumn{1}{c|}{26.95} & 53.98 & 38.84 & 28.19 \\
\multicolumn{1}{|c|}{} & \textit{(10.66)} & \textit{(5.82)} & \multicolumn{1}{c|}{\textit{(4.27)}} & \textit{(8.13)} & \textit{(6.19)} & \multicolumn{1}{c|}{\textit{(3.57)}} & \textit{(9.44)} & \textit{(5.68)} & \multicolumn{1}{c|}{\textit{(2.79)}} & \textit{(9.18)} & \textit{(4.89)} & \multicolumn{1}{c|}{\textit{(3.75)}} & \textit{(9.20)} & \textit{(4.55)} & \textit{(2.99)} \\
\multicolumn{1}{|c|}{Percentage} & 37 & 52 & \multicolumn{1}{c|}{45} & 70 & 64 & \multicolumn{1}{c|}{74} & 79 & 88 & \multicolumn{1}{c|}{97} &  &  & \multicolumn{1}{c|}{} &  &  &  \\ \hline
\multicolumn{16}{|c|}{} \\
\multicolumn{16}{|c|}{$RMSE \times 100$} \\
\multicolumn{1}{|c|}{OLS} & 20.59 & 14.35 & \multicolumn{1}{c|}{10.37} & 20.52 & 14.59 & \multicolumn{1}{c|}{10.27} & 20.73 & 14.76 & \multicolumn{1}{c|}{10.25} & 24.79 & 14.57 & \multicolumn{1}{c|}{10.36} & 20.95 & 14.70 & 10.37 \\
\multicolumn{1}{|c|}{} & \textit{(6.64)} & \textit{(3.39)} & \multicolumn{1}{c|}{\textit{(1.75)}} & \textit{(2.97)} & \textit{(1.51)} & \multicolumn{1}{c|}{\textit{(0.82)}} & \textit{(2.16)} & \textit{(0.92)} & \multicolumn{1}{c|}{\textit{(0.50)}} & \textit{(2.02)} & \textit{(0.85)} & \multicolumn{1}{c|}{\textit{(0.43)}} & \textit{(1.53)} & \textit{(0.75)} & \textit{(0.38)} \\
\multicolumn{1}{|c|}{GLS} & 15.27 & 10.01 & \multicolumn{1}{c|}{7.12} & 14.39 & 10.10 & \multicolumn{1}{c|}{7.10} & 14.59 & 10.25 & \multicolumn{1}{c|}{7.06} & 14.22 & 10.09 & \multicolumn{1}{c|}{7.16} & 14.52 & 10.19 & 7.14 \\
\multicolumn{1}{|c|}{} & \textit{(4.57)} & \textit{(2.27)} & \multicolumn{1}{c|}{\textit{(1.24)}} & \textit{(1.98)} & \textit{(1.07)} & \multicolumn{1}{c|}{\textit{(0.52)}} & \textit{(1.61)} & \textit{(0.71)} & \multicolumn{1}{c|}{\textit{(0.37)}} & \textit{(1.15)} & \textit{(0.54)} & \multicolumn{1}{c|}{\textit{(0.29)}} & \textit{(1.07)} & \textit{(0.55)} & \textit{(0.25)} \\
\multicolumn{1}{|c|}{FGLS} & 15.88 & 10.62 & \multicolumn{1}{c|}{7,51} & 17.80 & 12.57 & \multicolumn{1}{c|}{8.73} & 20.63 & 14.70 & \multicolumn{1}{c|}{10.22} &  &  & \multicolumn{1}{c|}{} &  &  &  \\
\multicolumn{1}{|c|}{} & \textit{(5.48)} & \textit{(2.56)} & \multicolumn{1}{c|}{\textit{(1.29)}} & \textit{(2.47)} & \textit{(1.43)} & \multicolumn{1}{c|}{\textit{(0.66)}} & \textit{(2.14)} & \textit{(0.89)} & \multicolumn{1}{c|}{\textit{(0.50)}} & \textit{} & \textit{} & \multicolumn{1}{c|}{\textit{}} & \textit{} & \textit{} & \textit{} \\
\multicolumn{1}{|c|}{FGLasso} & 16.47 & 10.99 & \multicolumn{1}{c|}{7.71} & 17.39 & 12.16 & \multicolumn{1}{c|}{8.21} & 18.32 & 12.55 & \multicolumn{1}{c|}{8.38} & 19.60 & 12.55 & \multicolumn{1}{c|}{8.60} & 18.95 & 12.79 & 8.68 \\
\multicolumn{1}{|c|}{} & \textit{(5.79)} & \textit{(2.49)} & \multicolumn{1}{c|}{\textit{(1.34)}} & \textit{(2.52)} & \textit{(1.43)} & \multicolumn{1}{c|}{\textit{(0.63)}} & \textit{(2.06)} & \textit{(0.78)} & \multicolumn{1}{c|}{\textit{(0.41)}} & \textit{(1.70)} & \textit{(0.74)} & \multicolumn{1}{c|}{\textit{(0.34)}} & \textit{(1.29)} & \textit{(0.66)} & \textit{(0.31)} \\
\multicolumn{1}{|c|}{Percentage} & 31 & 37 & \multicolumn{1}{c|}{42} & 66 & 75 & \multicolumn{1}{c|}{98} & 97 & 100 & \multicolumn{1}{c|}{100} &  &  & \multicolumn{1}{c|}{} &  &  &  \\ \hline
\multicolumn{1}{|l}{} & \multicolumn{1}{l}{} & \multicolumn{1}{l}{} & \multicolumn{1}{l}{} & \multicolumn{1}{l}{} & \multicolumn{1}{l}{} & \multicolumn{1}{l}{} & \multicolumn{1}{l}{} & \multicolumn{1}{l}{} & \multicolumn{1}{l}{} & \multicolumn{1}{l}{} & \multicolumn{1}{l}{} & \multicolumn{1}{l}{} & \multicolumn{1}{l}{} & \multicolumn{1}{l}{} & \multicolumn{1}{l|}{} \\
\multicolumn{16}{|c|}{Cross Validation} \\
\multicolumn{1}{|c|}{$\lambda_n \times 100$} & 6.34 & 4,83 & \multicolumn{1}{c|}{3.76} & 9.92 & 6.78 & \multicolumn{1}{c|}{7.78} & 17.93 & 15.81 & \multicolumn{1}{c|}{12.79} & 20.48 & 19.29 & \multicolumn{1}{c|}{16.77} & 24.29 & 21.90 & 18.55 \\
\multicolumn{1}{|c|}{} & \textit{(11.59)} & \textit{(8.89)} & \multicolumn{1}{c|}{\textit{(6.09)}} & \textit{(8.10)} & \textit{(4.82)} & \multicolumn{1}{c|}{\textit{(4.16)}} & \textit{(8.83)} & \textit{(5.93)} & \multicolumn{1}{c|}{\textit{(4.05)}} & \textit{(5.76)} & \textit{(4.91)} & \multicolumn{1}{c|}{\textit{(3.23)}} & \textit{(8.85)} & \textit{(4.36)} & \textit{(2.19)} \\ \hline
\end{tabular}}
\caption{AR(1)}
\label{tab4}
\begin{tablenotes}
\item \floatfoot{Notes: This table reports $\|\widehat{\beta}_{OLS}-\beta\|$, $\|\widehat{\beta}_{GLS}-\beta\|$, $\|\widehat{\beta}_{FGLS}-\beta\|$, $\|\widehat{\beta}_{FGLasso}-\beta\|$ by $l_\infty (\times 100)$, RMSE$(\times 100)$ when $\Omega$ is generated from AR(1) structure, see details in section 4.1. Percentage means the number of times $\|\widehat{\beta}_{FGLasso}-\beta\|\leq \|\widehat{\beta}_{FGLS}-\beta\|$ out of 100 simulations. $\lambda_n$ is the penalty parameter in (\ref{e.gl}) chosen by 5-fold cross validation. Values in parenthesis are the corresponding standard deviation$(\times 100)$. All results reported here are average of 100 replications.}
\end{tablenotes}
\end{threeparttable}
\end{table}

\begin{table}[htbp]
\centering
\begin{threeparttable}
\resizebox{\columnwidth}{!}{
\begin{tabular}{|cccccccccccccccc|}
\hline 
\multicolumn{1}{|c|}{} & \multicolumn{3}{c|}{N/T=0.1} & \multicolumn{3}{c|}{N/T=0.5} & \multicolumn{3}{c|}{N/T=1} & \multicolumn{3}{c|}{N/T=1.5} & \multicolumn{3}{c|}{N/T=2} \\ \cline{2-16} 
\multicolumn{1}{|c|}{T} & 50 & 100 & \multicolumn{1}{c|}{200} & 50 & 100 & \multicolumn{1}{c|}{200} & 50 & 100 & \multicolumn{1}{c|}{200} & 50 & 100 & \multicolumn{1}{c|}{200} & 50 & 100 & 200 \\
\multicolumn{1}{|c|}{N} & 5 & 10 & \multicolumn{1}{c|}{20} & 25 & 50 & \multicolumn{1}{c|}{100} & 50 & 100 & \multicolumn{1}{c|}{200} & 75 & 150 & \multicolumn{1}{c|}{300} & 100 & 200 & 400 \\ \hline
\multicolumn{16}{|c|}{} \\
\multicolumn{16}{|c|}{Element-wise maximum norm $l_\infty \times 100$} \\
\multicolumn{1}{|c|}{OLS} & 22.76 & 19.09 & \multicolumn{1}{c|}{15.09} & 32.28 & 24.79 & \multicolumn{1}{c|}{19.50} & 37.46 & 28.08 & \multicolumn{1}{c|}{21.21} & 38.10 & 29.97 & \multicolumn{1}{c|}{22.59} & 40.13 & 31.01 & 22.7 \\
\multicolumn{1}{|c|}{} & \textit{(8.16)} & \textit{(5.04)} & \multicolumn{1}{c|}{\textit{(3.22)}} & \textit{(7.30)} & \textit{(4.13)} & \multicolumn{1}{c|}{\textit{(2.97)}} & \textit{(7.57)} & \textit{(4.33)} & \multicolumn{1}{c|}{\textit{(2.81)}} & \textit{(5.84)} & \textit{(5.01)} & \multicolumn{1}{c|}{\textit{(2.58)}} & \textit{(6.07)} & \textit{(4.25)} & \textit{(2.6)} \\
\multicolumn{1}{|c|}{GLS} & 22.32 & 18.16 & \multicolumn{1}{c|}{14.68} & 31.09 & 24.01 & \multicolumn{1}{c|}{18.48} & 35.69 & 27.30 & \multicolumn{1}{c|}{20.37} & 36.57 & 28.64 & \multicolumn{1}{c|}{21.53} & 38.37 & 29.34 & 21.78 \\
\multicolumn{1}{|c|}{} & \textit{(7.96)} & \textit{(5.01)} & \multicolumn{1}{c|}{\textit{(2.97)}} & \textit{(7.49)} & \textit{(3.82)} & \multicolumn{1}{c|}{\textit{(2.72)}} & \textit{(6.98)} & \textit{(4.65)} & \multicolumn{1}{c|}{\textit{(2.63)}} & \textit{(5.94)} & \textit{(4.35)} & \multicolumn{1}{c|}{\textit{(2.36)}} & \textit{(6.11)} & \textit{(3.89)} & \textit{(2.61)} \\
\multicolumn{1}{|c|}{FGLS} & 23.64 & 19.14 & \multicolumn{1}{c|}{15.01} & 34.23 & 27.07 & \multicolumn{1}{c|}{20.9} & 37.84 & 28.33 & \multicolumn{1}{c|}{21.27} &  &  & \multicolumn{1}{c|}{} &  &  &  \\
\multicolumn{1}{|c|}{} & \textit{(8.50)} & \textit{(5.45)} & \multicolumn{1}{c|}{\textit{(3.16)}} & \textit{(7.65)} & \textit{(4.52)} & \multicolumn{1}{c|}{\textit{(2.98)}} & \textit{(7.66)} & \textit{(4.38)} & \multicolumn{1}{c|}{\textit{(2.89)}} & \textit{} & \textit{} & \multicolumn{1}{c|}{\textit{}} & \textit{} & \textit{} & \textit{} \\
\multicolumn{1}{|c|}{FGLasso} & 23.23 & 18.85 & \multicolumn{1}{c|}{14.79} & 32.68 & 25.04 & \multicolumn{1}{c|}{19.43} & 38.24 & 28.31 & \multicolumn{1}{c|}{21.10} & 38.80 & 30.23 & \multicolumn{1}{c|}{22.31} & 40.55 & 31.25 & 22.47 \\
\multicolumn{1}{|c|}{} & \textit{(8.39)} & \textit{(5.17)} & \multicolumn{1}{c|}{\textit{(3.10)}} & \textit{(7.13)} & \textit{(3.97)} & \multicolumn{1}{c|}{\textit{(2.82)}} & \textit{(8.17)} & \textit{(4.61)} & \multicolumn{1}{c|}{\textit{(2.85)}} & \textit{(6.40)} & \textit{(4.69)} & \multicolumn{1}{c|}{\textit{(2.34)}} & \textit{(6.36)} & \textit{(4.16)} & \textit{(2.56)} \\
\multicolumn{1}{|c|}{Percentage} & 59 & 55 & \multicolumn{1}{c|}{52} & 60 & 77 & \multicolumn{1}{c|}{70} & 57 & 50 & \multicolumn{1}{c|}{58} &  &  & \multicolumn{1}{c|}{} &  &  &  \\ \hline
\multicolumn{16}{|c|}{} \\
\multicolumn{16}{|c|}{$RMSE \times 100$} \\
\multicolumn{1}{|c|}{OLS} & 13.69 & 9.86 & \multicolumn{1}{c|}{6.92} & 14.08 & 9.97 & \multicolumn{1}{c|}{7.08} & 14.34 & 10.15 & \multicolumn{1}{c|}{7.06} & 14.32 & 10.03 & \multicolumn{1}{c|}{7.13} & 14.40 & 10.09 & 7.13 \\
\multicolumn{1}{|c|}{} & \textit{(4.61)} & \textit{(2.38)} & \multicolumn{1}{c|}{\textit{(1.09)}} & \textit{(2.08)} & \textit{(0.93)} & \multicolumn{1}{c|}{\textit{(0.54)}} & \textit{(1.63)} & \textit{(0.63)} & \multicolumn{1}{c|}{\textit{(0.40)}} & \textit{(1.25)} & \textit{(0.61)} & \multicolumn{1}{c|}{\textit{(0.30)}} & \textit{(1.02)} & \textit{(0.49)} & \textit{(0.24)} \\
\multicolumn{1}{|c|}{GLS} & 13.51 & 9.44 & \multicolumn{1}{c|}{6.71} & 13.61 & 9.53 & \multicolumn{1}{c|}{6.77} & 13.77 & 9.72 & \multicolumn{1}{c|}{6.75} & 13.79 & 9.60 & \multicolumn{1}{c|}{6.83} & 13.79 & 9.67 & 6.83 \\
\multicolumn{1}{|c|}{} & \textit{(4.48)} & \textit{(2.39)} & \multicolumn{1}{c|}{\textit{(1.03)}} & \textit{(2.08)} & \textit{(0.85)} & \multicolumn{1}{c|}{\textit{(0.51)}} & \textit{(1.52)} & \textit{(0.63)} & \multicolumn{1}{c|}{\textit{(0.35)}} & \textit{(1.20)} & \textit{(0.59)} & \multicolumn{1}{c|}{\textit{(0.28)}} & \textit{(1.04)} & \textit{(0.45)} & \textit{(0.23)} \\
\multicolumn{1}{|c|}{FGLS} & 14.35 & 9.85 & \multicolumn{1}{c|}{7.00} & 15.20 & 10.73 & \multicolumn{1}{c|}{7.66} & 14.47 & 10.19 & \multicolumn{1}{c|}{7.08} &  &  & \multicolumn{1}{c|}{} &  &  &  \\
\multicolumn{1}{|c|}{} & \textit{(4.76)} & \textit{(2.57)} & \multicolumn{1}{c|}{\textit{(1.10)}} & \textit{(2.32)} & \textit{(0.10)} & \multicolumn{1}{c|}{\textit{(0.56)}} & \textit{(1.60)} & \textit{(0.65)} & \multicolumn{1}{c|}{\textit{(0.40)}} & \textit{} & \textit{} & \multicolumn{1}{c|}{\textit{}} & \textit{} & \textit{} & \textit{} \\
\multicolumn{1}{|c|}{FGLasso} & 14.09 & 9.74 & \multicolumn{1}{c|}{6.87} & 14.26 & 9.96 & \multicolumn{1}{c|}{7.02} & 14.61 & 10.21 & \multicolumn{1}{c|}{7.03} & 14.46 & 10.10 & \multicolumn{1}{c|}{7.11} & 14.56 & 10.14 & 7.12 \\
\multicolumn{1}{|c|}{} & \textit{(4.75)} & \textit{(2.51)} & \multicolumn{1}{c|}{\textit{(1.10)}} & \textit{(2.11)} & \textit{(0.91)} & \multicolumn{1}{c|}{\textit{(0.53)}} & \textit{(1.66)} & \textit{(0.64)} & \multicolumn{1}{c|}{\textit{(0.38)}} & \textit{(1.30)} & \textit{(0.59)} & \multicolumn{1}{c|}{\textit{(0.30)}} & \textit{(1.04)} & \textit{(0.48)} & \textit{(0.2)} \\
\multicolumn{1}{|c|}{Percentage} & 61 & 50 & \multicolumn{1}{c|}{66} & 79 & 89 & \multicolumn{1}{c|}{100} & 46 & 51 & \multicolumn{1}{c|}{69} &  &  & \multicolumn{1}{c|}{} &  &  &  \\ \hline
\multicolumn{1}{|l}{} & \multicolumn{1}{l}{} & \multicolumn{1}{l}{} & \multicolumn{1}{l}{} & \multicolumn{1}{l}{} & \multicolumn{1}{l}{} & \multicolumn{1}{l}{} & \multicolumn{1}{l}{} & \multicolumn{1}{l}{} & \multicolumn{1}{l}{} & \multicolumn{1}{l}{} & \multicolumn{1}{l}{} & \multicolumn{1}{l}{} & \multicolumn{1}{l}{} & \multicolumn{1}{l}{} & \multicolumn{1}{l|}{} \\
\multicolumn{16}{|c|}{Cross Validation} \\
\multicolumn{1}{|c|}{$\lambda_n \times 100$} & 9.42 & 7.18 & \multicolumn{1}{c|}{5.64} & 16.55 & 12.50 & \multicolumn{1}{c|}{8.86} & 15.37 & 13.18 & \multicolumn{1}{c|}{10.56} & 16.32 & 15.48 & \multicolumn{1}{c|}{11.30} & 18.59 & 16.22 & 12.42 \\
\multicolumn{1}{|c|}{} & \textit{(8.33)} & \textit{(6.85)} & \multicolumn{1}{c|}{\textit{(4.56)}} & \textit{(10.03)} & \textit{(6.25)} & \multicolumn{1}{c|}{\textit{(3.16)}} & \textit{(10.96)} & \textit{(7.67)} & \multicolumn{1}{c|}{\textit{(2.56)}} & \textit{(11.77)} & \textit{(6.45)} & \multicolumn{1}{c|}{\textit{(2.56)}} & \textit{(11.54)} & \textit{(7.55)} & \textit{(2.00)}\\ \hline
\end{tabular}}
\caption{Dense}
\label{tab5}
\begin{tablenotes}
\item \floatfoot{Notes: This table reports $\|\widehat{\beta}_{OLS}-\beta\|$, $\|\widehat{\beta}_{GLS}-\beta\|$, $\|\widehat{\beta}_{FGLS}-\beta\|$, $\|\widehat{\beta}_{FGLasso}-\beta\|$ by $l_\infty (\times 100)$, RMSE$(\times 100)$ when $\Omega$ is generated from dense structure, see details in section 4.1. Percentage means the number of times $\|\widehat{\beta}_{FGLasso}-\beta\|\leq \|\widehat{\beta}_{FGLS}-\beta\|$ out of 100 simulations. $\lambda_n$ is the penalty parameter in (\ref{e.gl}) chosen by 5-fold cross validation. Values in parenthesis are the corresponding standard deviation$(\times 100)$. All results reported here are average of 100 replications.}
\end{tablenotes}
\end{threeparttable}
\end{table}
\end{document}